\documentclass[journal, onecolumn]{IEEEtran}

\usepackage{microtype}
\usepackage{graphicx}
\usepackage{booktabs} 
\usepackage{environ}%
\usepackage[CJKbookmarks=true,
            bookmarksnumbered=true,
            bookmarksopen=true,
            colorlinks=true,
            citecolor=blue,
            linkcolor=red,
            anchorcolor=red,
            urlcolor=blue
            ]{hyperref}
\usepackage{geometry}
\usepackage{amsmath}
\usepackage{amssymb}
\usepackage{mathtools}
\usepackage{amsthm}
\usepackage{multicol}
\usepackage{dsfont}

\usepackage[capitalize,noabbrev]{cleveref}

\theoremstyle{plain}
\newtheorem{theorem}{Theorem}[section]

\newtheorem{lemma}[theorem]{Lemma}

\theoremstyle{definition}

\theoremstyle{remark}
\newtheorem{remark}[theorem]{Remark}

\usepackage[textsize=tiny]{todonotes}

\title{On the Optimal Bounds  for Noisy  Computing}
\usepackage{babel}
\usepackage{verbatim}
\usepackage{mathtools}
\usepackage{bm}
\usepackage{natbib}
\usepackage{amsmath}
\usepackage{amssymb}
\usepackage{multicol}
\setcitestyle{authoryear,open={(},close={)}} 

\makeatletter

\newcommand{\orn}{\ensuremath{\mathsf{OR}} }
\newcommand{\maxn}{\ensuremath{\mathsf{MAX}} }
\newcommand{\sort}{\ensuremath{\mathsf{SORT}} }

\DeclareMathOperator*{\argmax}{arg\,max}

\DeclareMathOperator*{\ind}{\mathbbm{1}}

\usepackage{graphicx}

\usepackage{babel}
\usepackage{verbatim}
\usepackage{mathtools}
\usepackage{bm}
\usepackage{amsmath}
\usepackage{amssymb}
\usepackage{multicol}

\usepackage{multirow}

\usepackage{bm}
\usepackage{enumitem}

\usepackage{babel}
\usepackage{subcaption}

\usepackage{algorithm}%
\usepackage{algorithmicx}
\usepackage[noend]{algpseudocode}

\usepackage{amsthm}
\usepackage{bbm}
\usepackage[normalem]{ulem}

\theoremstyle{plain}

\usepackage{color}
\definecolor{cm}{RGB}{0,0,200}
\definecolor{purple}{RGB}{200,0,200}

\theoremstyle{plain}

\usepackage{caption}
\usepackage{subcaption}
\theoremstyle{definition}

\usepackage{color}
\definecolor{cm}{RGB}{0,0,200}
\definecolor{purple}{RGB}{200,0,200}

\NewEnviron{resizealign}{\sbox0{
    $\begin{matrix}\displaystyle\BODY\end{matrix}$}%
  \sbox1{$(\theequation)$}%
  \sbox2{\parbox{\dimexpr \wd0 + 2\wd1}%
    {\begin{align}\BODY\end{align}}}
  \noindent\resizebox{\hsize}{!}{\usebox2}%
}
 \title{ On the Optimal Bounds for Noisy Computing}
\author{Banghua Zhu,
\and Ziao Wang,
\and Nadim Ghaddar,
\and Jiantao Jiao,
\and Lele Wang
\thanks{Banghua Zhu is with the Department of Electrical Engineering and Computer Sciences, University of California, Berkeley, Berkeley, CA 94720, USA (email: banghua@berkeley.edu).}
\thanks{Ziao Wang is with the Department of Electrical and Computer Engineering, University of British Columbia, Vancouver, BC V6T1Z4, Canada (email: ziaow@ece.ubc.ca).}
\thanks{Nadim Ghaddar is with the Department of Electrical and Computer Engineering, University of California San Diego, La Jolla, CA 92093, USA, (email: nghaddar@ucsd.edu).}
\thanks{Jiantao Jiao is with the Department of Electrical Engineering and Computer Sciences, University of California, Berkeley, Berkeley, CA 94720, USA (email: jiantao@berkeley.edu).}
\thanks{Lele Wang is with the Department of Electrical and Computer Engineering, University of British Columbia, Vancouver, BC V6T1Z4, Canada (email: lelewang@ece.ubc.ca).}
\thanks{A shorter version of this work is accepted at the 2023 IEEE International Symposium on Information Theory.}}
\begin{document}
\maketitle




  
    

    
  
\begin{abstract}
  We revisit the problem of computing with noisy information considered in~\citet{feige1994computing}, which includes computing the \orn function from noisy queries, and computing the $\mathsf{MAX}, \mathsf{SEARCH}$, and $\mathsf{SORT}$ functions from noisy pairwise comparisons. For $K$ given elements, the goal is to correctly recover the desired function with probability at least $1-\delta$ when  the outcome of each query is flipped with probability $p$. We consider both the adaptive sampling setting where   
  each query can be adaptively designed based on past outcomes, and the non-adaptive sampling setting where the query cannot depend on past outcomes. The prior work provides tight  bounds on the worst-case query complexity in terms of the dependence on  $K$. 
  However, the upper and lower bounds do not match in terms of the dependence on $\delta$ and $p$. We improve the lower bounds for all the four functions under both adaptive and non-adaptive query models. Most of our lower bounds match the upper bounds up to constant factors when either $p$ or $\delta$ is bounded away from $0$, while the ratio between the best prior upper and lower bounds goes to infinity when $p\rightarrow 0$ or $p\rightarrow 1/2$. On the other hand, we also provide matching upper and lower bounds for the number of queries in expectation, improving both the upper and lower bounds  for the variable-length query model. 
\end{abstract}

\section{Introduction}

The problem of computing with noisy information  has been studied extensively since the seminal work~\citep{feige1994computing}, which considers four problems: 
\begin{itemize}
    \item Computing the \orn function  of $K$ bits  from noisy  observations of the  bits;
    \item Finding the largest (or top-$N$) element among $K$ real-valued elements from noisy pairwise comparisons;
    \item  Searching the rank of a new element in an ordered list of $K$ elements from noisy  pairwise  comparisons;
    \item  Sorting $K$ elements from noisy  pairwise  comparisons.
\end{itemize} 
\citet{feige1994computing} is based on a simple noise model where each observation goes through a binary symmetric channel $\mathsf{BSC}(p)$, i.e. for each observation, with probability  $p$ we  see its flipped outcome, and with probability $1-p$ we  see its true value. They provide upper and lower bounds for the query complexity in terms of the total number of elements $K$, the noise probability $p$, and the desired confidence level $\delta$ when adaptive querying is allowed. They  establish the optimal query complexity in terms of dependence with respect to $K$. However, the exact sample/query complexity with respect to all parameters $K, \delta$, and $p$ is still not fully understood.  
In this paper, we revisit the problem of 
computing  under noisy observations  in both the adaptive sampling and non-adaptive sampling  settings. We aim to close the gap in the existing bounds and illustrate the  difference in query complexity between adaptive sampling and non-adaptive sampling.

Taking the problem of   computing the \orn function of $K$ bits as an example, assume that there are $K$ bits $(X_1, \cdots, X_K)\in\{0, 1\}^K$. The \orn function   is defined as
\begin{align}\label{eqn:def_or}
    \mathsf{OR}(X_1,\cdots, X_K) = \begin{cases}1, & \text{ if } \exists k\in[K], X_k = 1 \\ 
    0, & \text{ otherwise.}
    \end{cases}
\end{align}

The question is simple when we can query each bit noiselessly. In this case, $K$ queries are both sufficient and necessary since it suffices to query each bit once. And thus there is no benefit in applying adaptive querying compared to non-adaptive querying.
When the observation of each query goes through a binary symmetric channel $\mathsf{BSC}(p)$,  we  ask two questions:
\begin{itemize}
    \item How many queries (samples) do we need  in the worst case to recover the true \orn function value of any given sequences $X_1,\cdots, X_K$ with probability at least $1-\delta$? 
    \item Can adaptive sampling do better than non-adaptive sampling when noise is present? 
\end{itemize}
\renewcommand{\arraystretch}{2}
\begin{table}[!htbp]
    \centering
\scalebox{1.1}{ 
    \begin{tabular}[!htbp]{|c|c|c|}
      \hline
      \multirow{2}{*}{\textbf{Problem}}   & \multicolumn{2}{|c|}{\textbf{Fixed Length, Adaptive Sampling}} \\\cline{2-3}
         & Upper Bound & Lower Bound 
         \\
         \hline 
         \orn & $\mathcal{O}(\frac{K\log(1/\delta)}{1-H(p)})$  \citep{feige1994computing} & $\Omega(\frac{K}{1-H(p)}+ \frac{K\log(1/\delta)}{D_{\mathsf{KL}}(p\|1-p)})$ (Thm~\ref{thm:or_adaptive_lower}) \\
         \maxn & $\mathcal{O}(\frac{K\log(1/\delta)}{1-H(p)})$ \citep{feige1994computing} & $\Omega(\frac{K}{1-H(p)}+ \frac{K\log(1/\delta)}{D_{\mathsf{KL}}(p\|1-p)})$ (Thm~\ref{thm:max_adaptive_lower}) \\
         $\mathsf{SEARCH}$ & $\mathcal{O}(\frac{\log(K/\delta)}{1-H(p)})$ \citep{feige1994computing} & $\Omega(\frac{\log(K)}{1-H(p)}+\frac{\log(1/\delta)}{D_{\mathsf{KL}}(p\|1-p)})$ (Thm~\ref{thm:search_adaptive_lower}) \\
         \sort & $\mathcal{O}(\frac{K\log(K/\delta)}{1-H(p)})$ \citep{feige1994computing} & $\Omega( \frac{K\log(K)}{1-H(p)}+\frac{K\log(K/\delta)}{D_{\mathsf{KL}}(p\|1-p)})$ (Thm~\ref{thm:sort_adaptive_lower}) \\
         \hline 
               \multirow{2}{*}{\textbf{Problem}}  & \multicolumn{2}{|c|}{\textbf{Fixed Length, Non-adaptive Sampling}}  \\\cline{2-3}
         & Upper Bound (Appendix~\ref{app:non-adaptive_upper}) & Lower Bound  
         \\
         \hline 
         \orn & $\mathcal{O}(\frac{K\log(K/\delta)}{1-H(p)})$ & $\Omega(\max(K,\frac{K\log(K)p}{1-H(p)}, \frac{K\log(K)}{\log((1-p)/p)}))$ (Thm~\ref{thm:or_non-adaptive_lower})\\
         \maxn & $\mathcal{O}(\frac{K^2\log(K/\delta)}{1-H(p)})$ & $\Omega(\frac{K^2}{1-H(p)}+\frac{K^2\log(1/\delta)}{D_{\mathsf{KL}}(p\|1-p)})$ (Thm~\ref{thm:max_non-adaptive_lower})\\
         $\mathsf{SEARCH}$ & $\mathcal{O}(\frac{K\log(1/\delta)}{1-H(p)})$ & $\Omega(\frac{K}{1-H(p)}+\frac{K\log(1/\delta)}{D_{\mathsf{KL}}(p\|1-p)})$ (Thm~\ref{thm:search_non-adaptive_lower})\\
         \sort & $\mathcal{O}(\frac{K^2\log(K/\delta)}{1-H(p)})$ & $\Omega(K^2+\frac{K^2\log(K)}{D_{\mathsf{KL}}(p\|1-p))})$   (Thm~\ref{thm:sort_non-adaptive_lower})\\
         \hline 
             \multirow{2}{*}{\textbf{Problem}}   & \multicolumn{2}{|c|}{\textbf{Variable Length, Adaptive Sampling}} \\\cline{2-3}
         & \multicolumn{2}{|c|}{Matching Bound (Thm~\ref{thm:variable_length}) }
         \\
         \hline 
         \orn & \multicolumn{2}{|c|}{$\Theta(\frac{K}{1-H(p)}+ \frac{K\log(1/\delta)}{D_{\mathsf{KL}}(p\|1-p)})$ }\\
         \maxn & \multicolumn{2}{|c|}{$\Theta(\frac{K}{1-H(p)}+ \frac{K\log(1/\delta)}{D_{\mathsf{KL}}(p\|1-p)})$} \\
         $\mathsf{SEARCH}$ & \multicolumn{2}{|c|}{$\Theta(\frac{\log(K)}{1-H(p)}+\frac{\log(1/\delta)}{D_{\mathsf{KL}}(p\|1-p)})$ } \\
         \sort & \multicolumn{2}{|c|}{$\Theta( \frac{K\log(K)}{1-H(p)}+\frac{K\log(K/\delta)}{D_{\mathsf{KL}}(p\|1-p)})$ } \\ \hline
             \end{tabular}
}
  \caption{Summary of query complexity bounds 
  of  $\mathsf{OR}$, $\mathsf{MAX}$, $\mathsf{SEARCH}$ and $\mathsf{SORT}$. Here, we assume $\delta<0.49$. 
  }
    \label{tab:summary}
    \vspace{-1.3em}
    \end{table}

In~\citet{feige1994computing}, an adaptive tournament algorithm is proposed that achieves worst-case query complexity $\mathcal{O}(K\log(1/\delta)/(1-H(p)))$, and a corresponding lower bound $\Omega(K\log(1/\delta)/\log((1-p)/p))$ for adaptive sampling is provided. A simple calculation tells  us that their ratio $ \log((1-p)/p)/(1-H(p))$ goes to infinity as $p\rightarrow 0$ or $p\rightarrow 1/2$, which indicates that there is still a    gap between the upper and the lower bounds when the  noise probability $p$  is near the point $0$ or $1/2$. This calls for tighter upper or lower bound for these cases. In our paper, we improve the lower bound to $\Omega(K/(1-H(p))+K\log(1/\delta)/D_{\mathsf{KL}}(p\|1-p))$, which matches the existing upper bound up to a constant factor when either $p$ or $\delta$ is bounded away from $0$.   

One may wonder how many samples are needed when each query is not allowed to depend on previous outcomes. We provide a lower bound $\Omega(\max(K,K\log(K)p/(1-H(p))), K\log(K)/\log((1-p)/p))$ for this non-adaptive setting. On the other hand, a repetition-based upper bound $\mathcal{O}(K\log(K/\delta)/(1-H(p)))$  matches the lower bound up to a constant factor when both $p$ and $\delta$ are bounded  from $0$. 

Similarly, we ask the same questions for computing $\mathsf{MAX}$, $\mathsf{SEARCH}$, and $\mathsf{SORT}$. We defer the definitions and discussions  of the problems to Section~\ref{sec:preliminary}. Here
we summarize the best upper and lower bounds for the problems considered in this paper in Table~\ref{tab:summary}, either from previous art or from this paper. 

Our lower bounds take the form of $f(K)/(1-H(p))+g(K,\delta)/D_{\mathsf{KL}}(p\|1-p)$. Here the first term does not depend on $\delta$ for $\delta<0.49$, representing the number of queries one must pay regardless of the target error probability $\delta$. The second term  grows logarithmically with $1/\delta$, representing the price to pay for smaller target error probability. Here we always have $D_{\mathsf{KL}}(p\|1-p)\gtrsim 1-H(p)$ for $p\in(0, 1)$, with $D_{\mathsf{KL}}(p\|1-p)\asymp  1-H(p)$ when $p$ is bounded away from $0$. Technically, the first term is usually from Fano's inequality, which gives better dependence on $p$ but worse dependence on $\delta$. The second term is from a $\mathsf{KL}$-divergence based lower bound, which gives better dependence on $\delta$ but worse dependence on $p$. 

We also extend our bounds from the \emph{fixed length} query model to \emph{variable length} query model (a.k.a. fixed budget and fixed confidence in the bandit literature~\citep{kaufmann2016complexity}). In the fixed length setting considered above, we ask for the worst-case deterministic number of queries required in the worst case to recover the true value with probability at least $1-\delta$. In the variable length setting, the number of queries can be random and dependent on the past outcomes. And we ask for the expected number of queries to recover the true value with probability at least $1-\delta$. We  discuss the results for  variable length in Section~\ref{app:variable_length}, where we give  matching upper and lower bounds with respect to all parameters for computing all four functions, improving over both existing upper and lower bounds and closing the gap. 


\subsection{Related Work}\label{app:related}
The problems of noisy computation have been studied extensively before and after \citep{feige1994computing}. However, most of the existing research work focuses on tightening the dependence on $K$ instead of $p$ and $\delta$,  or extending the  results to a more general framework where the  noise follows a generalized model that includes $\mathsf{BSC}$ channel as a special case. Although worst-case upper and lower bounds are provided for the generalized model, most of the lower bounds are based on instances where the noise does not follow a $\mathsf{BSC}$ model. Thus the lower bounds do not apply to our case. 
\paragraph{Noisy binary searching}
The noisy searching problem was first introduced by R\'enyi~\citep{Renyi1961} and Ulam~\citep{Ulam1976}  and further developed by Berlekamp~\citep{Berlekamp1964} and Horstein~\citep{Horstein1963} in the context of coding for channels with feedback. The noisy searching algorithm in~\citep{Burnashev1974} by Burnashev and Zigangirov can be seen as a specialization of Horstein's coding scheme, whereas the algorithms in~\citep{feige1994computing,pelc1989searching, Karp2007} can be seen as an adaptation of the binary search algorithm to the noisy setting.


The first tight lower bound for variable-length adaptive sampling in noisy searching is given by~\citep{Burn1976}.  
The recent concurrent work~\citep{gu2023optimal} improves the dependence on constant when $p$ is some constant that is bounded away from $0$ and $1/2$. Our lower bounds are based on a different proof using Le Cam's method. The results do not require $p$ to be bounded, but are worse in terms of the constant dependence. \citep{gu2023optimal} also provides matching upper bound that is tight even with the constant in the variable-length setting. 
We provide more discussions in Section~\ref{app:variable_length}.  Making the upper and lower bounds match in the fixed-length query model  still remains an important open problem. 

In terms of the bounds for non-adaptive sampling,  the gap between $\mathcal{O}(\log(K))$ for adaptive sampling and $\mathcal{O}(K)$  can be seen from the noiseless case when $p=0$~\citep{Renyi1961}.  Here we provide an improved bound for the noisy case that has explicit dependence on $p$ and $\delta$.

\paragraph{Noisy Sorting and max selection}
The noisy sorting and max (or Top-$N$) selection problems have been usually studied together (e.g.,~\citep{feige1994computing}) and later have been extended to  a more general setting known as active ranking~\citep{Mohajer2017,Falahatgar2017,Shah2018,Heckel2019,Agarwal2017}, where the  noise $p_{ij}$ for the comparison of a pair of elements $i$ and $j$ is usually unknown and different for different pairs. 
Other related but different settings for noisy sorting in the literature include the cases
when some distance metric for permutations is to be minimized (rather than the the probability of error)~\citep{Ailon2008,Braverman2009,Ailon2011adv,Negahban2012,Wauthier2013,Rajkumar2014,Shah2016,Shah2018,Mao2018}, 
when the noise for each pairwise comparison is not \emph{i.i.d.} and is determined by some noise model (e.g. the Bradley--Terry--Luce model\citep{BTL1952})~\citep{Ajtai2009,Negahban2012,Rajkumar2014,Chen2015,Chen2017,ren2018}, or when the ordering itself is restricted to some subset of all permutations~\citep{Jamieson2011,Ailon-Begleiter2011}.

For noisy sorting, the best upper and lower bounds have been provided in~\citep{feige1994computing, isit2022paper}, which give an upper bound $\mathcal{O}(K\log(K/\delta)/(1-H(p)))$ and lower bound $\Omega(K\log(K)/(1-H(p)) + \log(1/\delta)/D_{\mathsf{KL}}(p\|1-p))$. We tighten the lower bound to be $\Omega(K\log(K)/(1-H(p)) + K\log(K/\delta)/D_{\mathsf{KL}}(p\|1-p))$. On the other hand, \citep{gu2023optimal} shows that the query complexity is $(1+o(1))(K\log(K)/(1-H(p)) + K\log(K)/D_{\mathsf{KL}}(p\|1-p))$, which does not scale with  $\delta$. Our lower bound for fixed-length improves the dependence on $\delta$. We also make the bounds match up to constant factors for all parameters in the variable-length setting. However, making  the $\delta$ dependence tight in upper bound for fixed-length remains an open problem. In terms of the non-adaptive sampling scenario, we provide an upper and lower bound that matches in terms of the dependence on $K$, but is still loose when both $p$ and $\delta$ go to $0$ simultaneously.

For max selection, the best known lower bound for fixed-length adaptive sampling is $\Omega(K\log(1/\delta)/\log((1-p)/p))$~\citep{feige1994computing}, while our result makes it tight when either $p$ or $\delta$ goes to $0$ and provides matching bounds for variable-length setting. On the other hand, \citep{Mohajer2017, Shah2018} discuss the gap between adaptive sampling and non-adaptive sampling. However, the $\Omega(K^3\log(K))$ lower bound for non-adaptive sampling in \citep{Shah2018} does not apply to our case since it is based on a generalized model where the noise probability is different. In our case, $\mathcal{O}(K^2\log(K))$ is a natural upper bound. However, it is unclear whether our lower bound  $\Omega(K^2)$ can be improved to  $\Omega(K^2\log(K))$.

\paragraph{\orn and best arm identification}
The noisy computation of \orn has been first studied in the literature of circuit with noisy gates~\citep{dobrushin1977lower, dobrushin1977upper, von1956probabilistic, pippenger1991lower, gacs1994lower} and noisy decision trees~\citep{feige1994computing, evans1998average, reischuk1991reliable}. Different from the other three problems we consider, computing \orn
 does not rely on pairwise comparisons, but instead directly queries the values of the bits. This is also related to the rich literature of best arm identification, which queries real-valued arms and aims to identify the arm with largest value (reward)~\citep{bubeck2009pure, audibert2010best, garivier2016optimal, jamieson2014best, gabillon2012best,  kaufmann2016complexity}. Indeed, any best-arm identification algorithm can be converted to an \orn computation by first finding the maximum and then query the binary value of the maximum. This recovers the best existing upper bound $\mathcal{O}(K\log(1/\delta)/(1-H(p)))$ for computation of \orn under adaptive sampling scenario~\citep{audibert2010best}. However, the lower bound for best arm identification does not apply to our case, since \orn has a binary output, while the best arm identification problem requires  the arm index. And our lower bound for fixed-length adaptive sampling improves over the best known lower bound $\Omega(K\log(1/\delta)/\log((1-p)/p))$ ~\citep{feige1994computing}. We also provide matching bounds for variable-length.

 For non-adaptive sampling, \citep{reischuk1991reliable, gacs1994lower} provide a lower bound $\Omega(K\log(K)/\log((1-p)/p))$. Our lower bound $\Omega(K\log(K)p/(1-H(p)))$ is tighter than the current lower bound when $p\rightarrow 1/2$, but looser when $p\rightarrow 0$. Thus the tightest bound is a maximum of $K$ and the two lower bounds.

\subsection{Problem Definition and Preliminaries}\label{sec:preliminary}

The \orn function is defined in Equation~\eqref{eqn:def_or}. We define the rest of the problems here. Different from $\mathsf{OR}$, the $\mathsf{MAX}$, $\mathsf{SEARCH}$ and $\mathsf{SORT}$ problems are all based on   noisy pairwise comparisons. Concretely, assume we have $K$ distinct real-valued items $X_1,\cdots, X_K$. Instead of querying the exact  value of each element, we can only query a pair of elements and observe their noisy comparison. For any queried    pair $(i,j)$, we will observe a sample from $\mathsf{Bern}(1-p)$ if $X_i>X_j$, and a sample from $ \mathsf{Bern}(p)$ if $X_i<X_j$.   

We have $\maxn(X_1,\cdots, X_K)=\argmax_{i\in[K]}X_i$, $\sort(X_1,\cdots, X_K) = \sigma$, where $\sigma:[K]\mapsto [K]$ is the permutation function such that $X_{\sigma(1)}< X_{\sigma(2)}< \cdots< X_{\sigma(K)}$. And  $\mathsf{SEARCH}(X; X_1,\cdots,X_K) =  i$, where $i$ satisfies that $X_i<X<X_{i+1}$ with $X_0=-\infty$ and $X_{K+1}=+\infty$. In the $\mathsf{SEARCH}$ problem, we assume that the ordering of $X_1,\cdots, X_K$ is given, and we are interested where the position of a new $X$ is. Thus, in each round we compare the given $X$ and any of the elements $X_i$.

We are interested in the probability of exact recovery of the functions. 
We consider both adaptive sampling and non-adaptive sampling. In adaptive sampling, the sampling distribution at each round can be dependent on the past queries and observations. In non-adaptive sampling, the sampling distribution in each round has  to be determined at the beginning and cannot change with the ongoing queries or observations. Throughout the paper, we assume that the desired error probability $\delta$ satisfies $\delta<0.49$. We use the terms ``querying'' and ``sampling'' interchangeably. 

\section{Computing the \orn function}

In this section,  we  provide the  lower bounds for the query complexity of computing the \orn function under both adaptive  and non-adaptive sampling. The upper bound for adaptive sampling is from~\citep{feige1994computing}. And the upper bound for non-adaptive sampling is omitted . Let $\theta\in\{0, 1\}^K$ be the $K$-bit sequence representing the true values. Let $\mathsf{OR}(\theta)$ be the result of the \orn function applied to the $K$-bit noiseless sequence. We also let $\hat \mu$ be any algorithm that queries any noisy bit in $T$ rounds, and outputs a (possibly random) decision from $\{0, 1\}$. 

\subsection{Adaptive Sampling}
 We have the following minimax lower bound. 
\begin{theorem}\label{thm:or_adaptive_lower} In the adaptive setting, we have
\begin{align*}
    \inf_{\hat\mu} \sup_{\theta\in\{0, 1\}^K} \mathbb{P}(\hat\mu \neq \mathsf{OR}(\theta))\geq \frac{1}{4}\cdot \exp\left(-\frac{T\cdot D_{\mathsf{KL}}(p\|1-p)}{K }\right).
\end{align*}
Thus, the number of queries required to recover the true value with probability at least  $1-\delta$ is lower bounded by $\Omega({K}/{(1-H(p))}+K\log(1/\delta)/D_{\mathsf{KL}}(p\|1-p))$. 
\end{theorem}

We provide the proof here, which is based on Le Cam's two point method (see e.g. \citep{lecam1973convergence, yu1997assouad}). 
\begin{proof}[Proof of Theorem~\ref{thm:or_adaptive_lower}]
Our lower bound proof is mainly based on Le Cam's two point method, which is also re-stated in  Lemma~\ref{lem:lecam} in Appendix~\ref{app:lemma} for reader's convenience. 
Let $\theta_0$ {be the length-$K$ all-zero sequence, and let $\theta_j \in \{0,1\}^K$ be such that $\theta_{jj} = 1$ and $\theta_{ji} = 0$ for $i \neq j$}. Here $\theta_{ji}$ refers to the $i$-th element in the binary vector $\theta_j$. We can first verify that for any $\hat\mu$, one has
\begin{align*}
    \mathds{1}(\hat\mu \neq \mathsf{OR}(\theta_0)) +  \mathds{1}(\hat\mu \neq \mathsf{OR}(\theta_j)) \geq 1.
\end{align*}
By applying Le Cam's two point lemma on $\theta_0$ and $\theta_j$, we know that
\begin{align*}
      \inf_{\hat\mu} \sup_{\theta\in\{0, 1\}^K} \mathbb{P}(\hat\mu \neq \mathsf{OR}(\theta)) & \geq \frac{1}{2}(1-\mathsf{TV}(\mathbb{P}_{\theta_0}, \mathbb{P}_{\theta_j})). 
\end{align*}
Here $\mathbb{P}_{\theta_j}$ is the joint distribution of query-observation pairs in $T$ rounds when the ground truth is $\theta_j$. By taking maximum over $j$ on the right-hand side, we have
\begin{align*}
     & \inf_{\hat\mu} \sup_{\theta\in\{0, 1\}^K} \mathbb{P}(\hat\mu \neq \mathsf{OR}(\theta))  \\
     & \geq \sup_{1\leq j\leq K}\frac{1}{2}(1-\mathsf{TV}(\mathbb{P}_{\theta_0}, \mathbb{P}_{\theta_j}))  \\ 
      & \geq \sup_{1\leq j\leq K} \frac{1}{4}\exp(-D_{\mathsf{KL}}(\mathbb{P}_{\theta_0}, \mathbb{P}_{\theta_j})).
\end{align*}
Here the last inequality is due to Bretagnolle–Huber inequality~\citep{bretagnolle79} (Lemma~\ref{lem:bh} in Appendix~\ref{app:lemma}). Now we aim at computing $D_{\mathsf{KL}}(\mathbb{P}_{\theta_0}, \mathbb{P}_{\theta_j})$. Let $T_j$ be the random variable that denotes the number of times the $j$-th element is queried among all $T$ rounds.  From divergence decomposition lemma~\citep{auer1995gambling} (Lemma~\ref{lem:div} in Appendix~\ref{app:lemma}), we have  
\begin{align*}
   D_{\mathsf{KL}}(\mathbb{P}_{\theta_0}, \mathbb{P}_{\theta_j}) = \mathbb{E}_{\theta_0}[T_j] \cdot D_{\mathsf{KL}}(p\|1-p).
\end{align*}
Here $\mathbb{E}_{\theta_0}[T_j]$ denotes the expected number of times the $j$-th element is queried when the ground truth is $\theta_0$.
Thus we have
\begin{align*}
     & \inf_{\hat\mu} \sup_{\theta\in\{0, 1\}^K} \mathbb{P}(\hat\mu \neq \mathsf{OR}(\theta))  \\
      & \geq \sup_{1\leq j\leq K} \frac{1}{4}\exp(-\mathbb{E}_{\theta_0}[T_j] \cdot D_{\mathsf{KL}}(p\|1-p)).
\end{align*}

Now since $\sum_j \mathbb{E}_{\theta_0}[T_j] = T$, there must exists some $j$ such that $ \mathbb{E}_{\theta_0}[T_j]  \leq T/K$. This gives
\begin{align*}
      \inf_{\hat\mu} \sup_{\theta\in\{0, 1\}^K} \mathbb{P}(\hat\mu \neq \mathsf{OR}(\theta)) &  \geq  \frac{1}{4}\cdot \exp\left(-\frac{T\cdot D_{\mathsf{KL}}(p\|1-p)}{K }\right).
\end{align*}

 On the other hand, $K$ is naturally a lower bound for query complexity since one has to query each element at least once. Thus we arrive at a lower bound of $\Omega(K+K\log(1/\delta)/D_{\mathsf{KL}}(p\|1-p))$. Note that this is equivalent to $\Omega(K/(1-H(p))+K\log(1/\delta)/D_{\mathsf{KL}}(p\|1-p))$ up to a constant factor when $\delta<0.49$. The reason is that when $p$ is bounded away from $0$,  $(1-H(p))/D_{\mathsf{KL}}(p\|1-p)$ is always some constant. When $p$ is close to $0$, $1-H(p)$ is within constant factor of $1$. 
\end{proof}

\begin{remark}\label{rmk:or}

Compared with the existing tightest bound $\Omega(K\log(1/\delta)/\log((1-p)/p))$ in~\citet{feige1994computing}, the rate is greatly improved as $p\rightarrow 0$ or $p\rightarrow 1/2$. 

On the other hand, the best known upper bound from~\citet{feige1994computing}, which is $\mathcal{O}(\frac{K\log(1/\delta)}{1-H(p)})$. We include its algorithm and analysis in Appendix~\ref{app:upper_or}. 
Theorem~\ref{thm:or_adaptive_lower}
shows that  when $\delta$ is bounded away from $0$, one needs at least $ C
\cdot K/(1-H(p))$ samples, matching the upper bound. 
Similarly, when $p$ is bounded away from $0$, the term $D_{\mathsf{KL}}(p\|1-p)$ is also within a constant factor of $1-H(p)$, thus the upper and lower bounds match.  The only regime where the upper and lower bounds do not match is the case when both $p$ and $\delta$ go to $0$. We conjecture that a better upper bound is needed in this case.
\end{remark}

\subsection{Non-adaptive Sampling}

In the case of non-adaptive sampling, 
we show that $\mathcal{O}(K)$ queries are not enough. And one needs $\Omega(K\log(K))$ queries.   

\begin{theorem}\label{thm:or_non-adaptive_lower}
In the non-adaptive sampling setting, where  the sampling procedure is restricted to taking independent samples from 
a sequence of distributions $p_1,\cdots, p_T$, we have
\begin{align*}
& \inf_{\hat\mu} \sup_{\theta\in\{0, 1\}^K} \mathbb{P}(\hat\mu \neq \mathsf{OR}(\theta))\\
& \quad \geq \tfrac{1}{2}\cdot \Big(1-\sqrt{\tfrac{1}{2K}\Big(\left(1+\tfrac{(1-2p)^2}{K(1-p)p}\right)^{T}-1\Big)}\Big).
\end{align*}
This shows that  the query complexity is at least $\Omega(\max(K,K\log(K)p/(1-H(p))))$. 
\end{theorem} 
We provide the proof below. Different from the case of adaptive sampling, for non-adaptive sampling we target for a rate of $\Omega(K\log(K))$. And a standard Le Cam's two point method is not sufficient to give the extra logarithmic factor. Thus we provide a new proof  based on a point versus mixture extension of Le Cam's method.

\begin{proof}[Proof of Theorem
~\ref{thm:or_non-adaptive_lower}]
Consider the instance $0$ which has all $0$ as its elements. For instance $1$, we define it as a  distribution $q$ over $K$ instances $1_k$, $k\in[K]$ which puts probability  $p_k$ on the $k$-th element. We will determine the value of $p_k$  later. Here instance  $1_k$ refers  to the case when $k$-th element is $1$, and the rest elements are $0$.  Now from Le Cam's two point lemma, we have
\begin{align*}
      & \inf_{\hat\mu} \sup_{\theta\in\{0, 1\}^K} \mathbb{P}(\hat\mu \neq \mathsf{OR}(\theta))
      \\
      & \geq \frac{1}{2}(1-\mathsf{TV}(\mathbb{P}_{1}, \mathbb{P}_{0}))  \\
      & = \frac{1}{2}(1-\mathsf{TV}(\mathbb{E}_{j\sim q}[\mathbb{P}_{1_j}], \mathbb{P}_{0})) \\ 
      & \geq \frac{1}{2}\left(1-\sqrt{\frac{1}{2}\chi^2(\mathbb{E}_{j\sim q}[\mathbb{P}_{1_j}], \mathbb{P}_{0})}\right).
\end{align*}
Here the last inequality is based on the inequality $\mathsf{TV}\leq \sqrt{\frac{1}{2}\chi^2}$. Let $\pi_m^t$ be the probability of querying the $m$-th element in round $t$. 
We can further calculate the $\chi^2$ divergence as
\begin{align*}
& \chi^2(\mathbb{E}[\mathbb{P}_{1_j}(\cdot )], \mathbb{P}_{0}(\cdot )) \\ 
& \stackrel{(a)}{=}\mathbb{E}_{j,j'\sim q} \left[\sum_x \frac{\mathbb{P}_{1_j}(x)\mathbb{P}_{1_{j'}}(x)}{\mathbb{P}_{0}(x)} \right] -1 \\
& \stackrel{(b)}{=} \mathbb{E}_{j,j'} \Bigg[\prod_{t=1}^T\Bigg(\sum_{m=1}^K \frac{(\pi^t_m)^2(1-p)^{\ind(j\neq m)+\ind(j'\neq m)}p^{\ind(j= m)+\ind(j'= m)}}{\pi_m^t(1-p)}  \\
&\hspace{2em}+\frac{(\pi^t_m)^2(1-p)^{\ind(j= m)+\ind(j'=m)}p^{\ind(j\neq m)+\ind(j'\neq m)}}{\pi_m^t p}  \Bigg)\Bigg] -1\\ 
& = \mathbb{E}_{j,j'} \Bigg[\prod_{t=1}^T\Bigg(\sum_{m=1}^K \pi^t_m \Bigg(\frac{(1-p)^{\ind(j\neq m)+\ind(j'\neq m)}p^{\ind(j= m)+\ind(j'= m)}}{1-p} \\
&\hspace{2em}+  \frac{(1-p)^{\ind(j= m)+\ind(j'=m)}p^{\ind(j\neq m)+\ind(j'\neq m)}}{ p}  \Bigg)\Bigg)\Bigg] -1 \\
& =  \mathbb{E}_{j,j'} \left[\prod_{t=1}^T\left(\sum_{m=1}^K \pi_m^t \left( 1+\frac{(1-2p)^2 \cdot \ind(j=j'=m)}{(1-p)p}\right)\right)\right]-1 \\
& = 1 - \sum_{j=1}^K p_j^2 +   \sum_{j=1}^K p_j^2\prod_{t=1}^T\left(1+\frac{\pi_j^t(1-2p)^2}{(1-p)p}\right) -1\\ 
& =  - \sum_{j=1}^K p_j^2 +   \sum_{j=1}^K p_j^2\prod_{t=1}^T\left(1+\frac{\pi_j^t(1-2p)^2}{(1-p)p}\right),
\end{align*}
where $(a)$ follows from Lemma~\ref{lem:chisquare}, and $(b)$ follows from the tensorization property of $\chi^2$ for tensor products, a direct result of Lemma~\ref{lem:chisquare}.
Now denote $T_j = \sum_{t=1}^T \pi_j^t$. By Jensen's inequality, we know that $\prod_{t=1}^T(1+\frac{\pi_j^t(1-2p)^2}{(1-p)p})\leq(1+\frac{T_j(1-2p)^2}{T(1-p)p})^T$. Thus we have
\begin{align*}
& \chi^2(\mathbb{P}_{0}(\cdot ), \mathbb{E}[\mathbb{P}_{1_j}(\cdot )]) \\
   & \leq  - \sum_{j=1}^K p_j^2 +   \sum_{j=1}^K p_j^2 \left(1+\frac{T_j(1-2p)^2}{T(1-p)p}\right)^T. 
\end{align*}
Now we take $p_j = ((1+\frac{T_j(1-2p)^2}{T(1-p)p})^T-1)^{-1/2}/(\sum_j ((1+\frac{T_j(1-2p)^2}{T(1-p)p})^T-1)^{-1/2})$. We have
\begin{align*}
    & \chi^2(\mathbb{P}_{0}(\cdot ), \mathbb{E}[\mathbb{P}_{1_j}(\cdot )]) \\
   & \leq  \frac{K}{(\sum_{j=1}^K ((1+\frac{T_j(1-2p)^2}{T(1-p)p})^T-1)^{-1/2})^2}.
\end{align*}
Since $\sum_j T_j = T$,
by Jensen's inequality, we have $\sum_{j=1}^K ((1+\frac{T_j(1-2p)^2}{T(1-p)p})^T-1)^{-1/2} \geq K ((1+\frac{(1-2p)^2}{K(1-p)p})^T-1)^{-1/2}$. Thus
\begin{align*}
\chi^2(\mathbb{P}_{0}(\cdot ), \mathbb{E}[\mathbb{P}_{1_j}(\cdot )])  \leq  \frac{1}{K}\cdot \left(\left(1+\frac{(1-2p)^2}{K(1-p)p}\right)^T-1\right).
\end{align*}
This gives the desired result. Now solving the inequality
\begin{align*}
\delta \geq \frac{1}{2}\cdot \left(1-\sqrt{\frac{1}{2K}\left(\left(1+\frac{(1-2p)^2}{K(1-p)p}\right)^{T}-1\right)}\right),
\end{align*}
we arrive at $T\geq \log(1+2K(1-2\delta^2))/\log(1+\frac{(1-2p)^2}{K(1-p)p})\gtrsim K\log(K)p/(1-2p)^2$.
\end{proof}

\begin{remark}
 Compared with the   bound  $\Omega(\frac{K\log(K)}{\log((1-p)/p)})$ in \citep{reischuk1991reliable, gacs1994lower}, our lower bound is tighter when $p\rightarrow 1/2$, but looser when $p\rightarrow 0$. Thus the tightest lower bound is a maximum of $K$ and the two  lower bounds. The corresponding repetition-based upper bound $\mathcal{O}(\frac{K\log(K/\delta)}{1-H(p)})$ can be derived by a union-bound based argument. 
    Compared with the upper bound $\mathcal{O}(\frac{K\log(K/\delta)}{1-H(p)})$, the lower bound is tight with respect to all parameters when both $p$ and $\delta$ are bounded away from $0$.  
\end{remark}

\section{Computing the $\mathsf{MAX}$ function}


Let $\theta\in[0, 1]^K$ be a sequence of length $K$ representing the true values of each element. $\mathsf{MAX}(\theta)$ be the index of the maximum value in the sequence. We also let $\hat \mu$ be any algorithm that (possibly randomly) queries any noisy comparison between two elements in $T$ rounds, and output a (possibly random) decision from $0, 1$. 
\subsection{Adaptive Sampling} 
We have the following minimax lower bound for the adaptive setting.  The proof is deferred to Appendix~\ref{proof:max_adaptive}.
\begin{theorem}\label{thm:max_adaptive_lower}In the adaptive setting, we have 
\begin{align*}
    \inf_{\hat\mu} \sup_{\theta\in[0, 1]^K} \mathbb{P}(\hat\mu \neq \mathsf{MAX}(\theta))\geq \tfrac{1}{2}\cdot \exp\left(-\tfrac{T\cdot D_{\mathsf{KL}}(p\|1-p)}{K }\right).
\end{align*}
Thus, the number of queries required to recover the true value with probability at least  $1-\delta$ is lower bounded by $\Omega({K}/{(1-H(p))}+K\log(1/\delta)/D_{\mathsf{KL}}(p\|1-p))$. 
\end{theorem}

The comparison with the existing lower bound~\citep{feige1994computing} for computing \maxn is similar as that for \orn in Remark~\ref{rmk:or}, and is thus omitted here.

\subsection{Non-adaptive Sampling}
In the case of non-adaptive sampling, we show that $\mathcal{O}(K)$ samples are not enough. Instead, one needs at least $\Omega(K^2)$ samples. We have the following result. The proof is deferred to Appendix~\ref{proof:max_non-adaptive}.

\begin{theorem}\label{thm:max_non-adaptive_lower}In the non-adaptive setting, where the sampling procedure is restricted to taking independent samples from 
a sequence of distributions $p_1,\cdots, p_T$, we have
\begin{align*}
\inf_{\hat\mu} \sup_{\theta\in[0, 1]^K} \mathbb{P}(\hat\mu \neq \mathsf{MAX}(\theta))\geq \tfrac{1}{4}\cdot \exp\left(-\tfrac{T\cdot D_{\mathsf{KL}}(p\|1-p)}{K^2 }\right).
\end{align*} 
Thus the queries required to recover the true value with probability at least  $1-\delta$ is lower bounded by $\Omega(K^2/(1-H(p))+K^2\log(1/\delta)/D_{\mathsf{KL}}(p\|1-p))$. 
\end{theorem}

\begin{remark} 
Compared with the repetition-based upper bound $\mathcal{O}(K^2\log(K))$, the lower bound has a $\log(K)$ gap.  The tight dependence on $K$ remains open.
\end{remark}

\section{Computing the $\mathsf{SEARCH}$ function}

\subsection{Adaptive Sampling}
Recall that for any sorted sequence $X_1,\cdots, X_K$, the $\mathsf{SEARCH}$ function for $X$ is defined as  $\mathsf{SEARCH}(X; X_1,\cdots,X_K) =  i$, where $i$ satisfies $X_i<X<X_{i+1}$. We start with adaptive setting. The proof is deferred to Appendix~\ref{proof:search_adaptive}.

\begin{theorem}\label{thm:search_adaptive_lower}In the adaptive setting, we have
\begin{align*}
& \inf_{\hat\mu} \sup_{X} \mathbb{P}(\hat\mu \neq \mathsf{SEARCH}(X)) \\ 
    \geq & \tfrac{1}{4}\cdot \max\Big( \exp\left(-{T\cdot D_{\mathsf{KL}}(p\|1-p)}\right), 1-\tfrac{T\cdot (1-H(p))}{\log(K)}\Big).
\end{align*} 
Thus the queries required to recover the true value with probability at least  $1-\delta$ is lower bounded by $\Omega(\log(K)/(1-H(p))+\log(1/\delta)/D_{\mathsf{KL}}(p\|1-p))$. 
\end{theorem}
\begin{remark}
The same lower bound is also proven in two different manners in the concurrent work~\citep{isit2022paper, gu2023optimal}. And a matching upper bound that is tight with respect to all parameters when $p$ is bounded away from $0$ and $1/2$ is provided in~\citet{gu2023optimal}.
\end{remark}
\subsection{Non-adaptive Sampling}
For non-adaptive sampling, we show $\Omega(K)$ queries are needed. The proof is deferred to Appendix~\ref{proof:search_non-adaptive}. 
\begin{theorem}\label{thm:search_non-adaptive_lower}In the non-adaptive setting where   the sampling procedure is restricted to taking independent samples from 
a sequence of distributions $p_1,\cdots, p_T$, we have
\begin{align*}
\inf_{\hat\mu} \sup_{X} \mathbb{P}(\hat\mu \neq \mathsf{SEARCH}(X))\geq \tfrac{1}{4}\cdot \exp\left(-\tfrac{T\cdot D_{\mathsf{KL}}(p\|1-p)}{K}\right).
\end{align*} 
Thus, the queries required to recover the true value with probability at least  $1-\delta$ is lower bounded by $\Omega(K/(1-H(p))+K\log(1/\delta)/D_{\mathsf{KL}}(p\|1-p))$. 
\end{theorem}

\begin{remark} We show in Appendix~\ref{app:non-adaptive_upper} the upper bound $\mathcal{O}(K\log(1/\delta)/(1-H(p)))$ for computing $\mathsf{SEARCH}$.  Our lower bound matches with all parameters when either $p$ or $\delta$ is bounded away from $0$. 
\end{remark}

\section{Computing the $\mathsf{SORT}$ function}

\subsection{Adaptive Sampling}
Let $\theta\in[0,1]^K$ be any sequence of $K$ elements with distinct values in $[0,1]$, $\mathsf{SORT}(\theta)$ be the result of the $\mathsf{SORT}$ function applied on the noiseless sequence. We also let $\hat \mu$ be any algorithm that queries any noisy comparison between two elements in $T$ rounds, and output a (possibly random) decision.  We have the following result for computing the $\mathsf{SORT}$ function. The proof is deferred to Appendix~\ref{proof:sort_adaptive}. 
\begin{theorem}\label{thm:sort_adaptive_lower}In the adaptive setting, we have 
\begin{align*}
    & \inf_{\hat\mu} \sup_{\theta\in[0, 1]^K} \mathbb{P}(\hat\mu \neq \mathsf{SORT}(\theta))\\ 
    \geq & \frac{1}{4}\max\left( \exp\left(-\frac{T\cdot D_{\mathsf{KL}}(p\|1-p)}{K}\right), 1-\frac{T\cdot (1-H(p))}{K\log(K)}\right).
\end{align*}
Thus, the queries required to recover the true value with probability at least  $1-\delta$ is lower bounded by $\Omega(K\log(K)/(1-H(p))+K\log(K/\delta)/D_{\mathsf{KL}}(p\|1-p)))$.
\end{theorem}
\begin{remark}
    Compared with the upper bound $\mathcal{O}(K\log(K/\delta)/(1-H(p)))$, the bound is tight with all parameters when  either $p$ or $\delta$ is bounded away from $0$. 
\end{remark}
\subsection{Non-adaptive Sampling}
Here we provide the following minimax lower bound for non-adaptive learning. The proof is deferred to Appendix~\ref{proof:sort_non-adaptive}. 
\begin{theorem}\label{thm:sort_non-adaptive_lower}In the non-adaptive setting where   the sampling procedure is restricted to taking independent samples from 
a sequence of distributions $p_1,\cdots, p_T$, 
\begin{align*}
    \inf_{\hat\mu} \sup_{\theta\in[0, 1]^K} \mathbb{P}(\hat\mu \neq \mathsf{SORT}(\theta))\geq  \frac{1}{2}\cdot \left(1-\frac{T\cdot D_{\mathsf{KL}}(p\|1-p)}{K^2\log(K)}\right).
\end{align*}
Thus, the queries required to recover the true value with probability at least  $1-\delta$ is lower bounded by $\Omega(\max(K^2,K^2\log(K)/D_{\mathsf{KL}}(p\|1-p)))$.
\end{theorem}
\begin{remark}
    Compared with the repetition-based upper bound $\mathcal{O}(K^2\log(K/\delta)/(1-H(p)))$, the lower bound is tight with all parameters when  $p$ and $\delta$ are bounded away from $0$. 
\end{remark}

\section{Matching Bounds for Variable Length}\label{app:variable_length}

In this section, we provide matching upper and lower bounds for the variable-length setting. All the bounds here are the same as the lower bound for the fixed-length setting, the proof of which can be directly adapted from the fixed-length results.  

\begin{theorem}\label{thm:variable_length} In the adaptive setting, the number of queries  in expectation to achieve at most $\delta$ error probability is
\begin{enumerate}
    \item $\Theta(\frac{K}{1-H(p)}+\frac{K\log(1/\delta)}{D_{\mathsf{KL}}(p\|1-p)})$  for computing \orn\!;
      \item $\Theta(\frac{K}{1-H(p)}+\frac{K\log(1/\delta)}{D_{\mathsf{KL}}(p\|1-p)})$  for computing $\mathsf{MAX}$;
    \item $\Theta(\frac{\log(K)}{1-H(p)}+\frac{\log(1/\delta)}{D_{\mathsf{KL}}(p\|1-p)})$  for computing $\mathsf{SEARCH}$;
    \item $\Theta(\frac{K\log(K)}{1-H(p)}+\frac{K\log(K/\delta)}{D_{\mathsf{KL}}(p\|1-p)})$  for computing $\mathsf{SORT}$.
\end{enumerate}
\end{theorem}

The proof is deferred to Appendix~\ref{proof:variable_length}. 
The matching upper bound for $\mathsf{SEARCH}$ is given in   \citep{gu2023optimal} in the regime when $p$ is some constant that is bounded away from $0$ and $1/2$, where they make it tight even for the  dependency on the constant. Our results for the other three functions improve both existing upper and lower bounds, and provide tight query complexity with all parameters in the variable-length setting. \citep{isit2022paper} initiates the study for constant-wise matching  bounds for $\mathsf{SORT}$ when $\delta = \Omega(K^{-1})$, which is achieved by~\citep{gu2023optimal}. Our bounds are tight up to constant for arbitrarily small $\delta$.

\begin{algorithm}[!htbp]
\caption{Variable-length tournament for computing \orn with noise}
\label{alg:tournament_variable}
  \begin{algorithmic}[1]
\State  \textbf{Input}: Target confidence level $\delta$.
  \State Set $\mathcal{X} = (X_1,X_2,\cdots, X_K)$ as the list of all bits with unknown value.
\For{iteration $i= 1:\lceil\log_2(K)\rceil $}

 \For{iteration $j = 1:\lceil |\mathcal{X}|/2\rceil $}
 \State Set $a = 1/2$, $\tilde\delta_i = \delta^{2(2i-1)}$.
 \While{$a\in(\tilde\delta_i, 1-\tilde\delta_i)$}
     \State Query the $(2j-1)$-th element once. 
     \State If observe $1$, update $a = \frac{(1-p)a}{(1-p)a + p(1-a)}$. Otherwise update  $a = \frac{pa}{pa + (1-p)(1-a)}$.  
     \EndWhile
      \State If $a\leq \tilde\delta_i$, remove the $(2j-1) $-th element, otherwise remove the $2j$-th element.
      \EndFor
      \State Break when $\mathcal{X}$ only has one element left. 
      \EndFor
 \State Set $a=1/2$. 
 \While{$a\in(\delta, 1-\delta)$}
     \State Query the only left element in $\mathcal{X}$ element once. 
     \State If observe $1$, update $a = \frac{(1-p)a}{(1-p)a + p(1-a)}$. Otherwise update  $a = \frac{pa}{pa + (1-p)(1-a)}$.  
     \EndWhile
      \State If $a\leq \delta$, return $0$, otherwise return $1.$
\end{algorithmic}
  \end{algorithm}
We provide the upper bound algorithm for computing \orn in Algorithm~\ref{alg:tournament_variable}. 
For the upper bound, one major difference is that to compare two elements with error  probability at most $\delta$, one needs  $\mathcal{O}(\log(1/\delta)/D_{\mathsf{KL}}(p\|1-p)+1/(1-2p))$ queries in expectation, which can be achieved by keep comparing the two elements until the posterior distribution reaches the desired confidence level (see e.g. Lemma 13 of \citet{gu2023optimal}). But the best known bound for fixed-length  is $\mathcal{O}(\log(1/\delta)/(1-H(p))$~\citep{feige1994computing}. This makes it  simpler to achieve tight rate for variable length. 

In Algorithm~\ref{alg:tournament_variable}, we adapt the noisy comparison oracle in~\citet{gu2023optimal} to noisy query oracle on each element, and combine with the original fixed-length algorithm in~\citet{feige1994computing}. In each round, the algorithm eliminates half of the elements in the current set by querying the elements with odd indices. If the $(2j-1)$-th element is determined to be $1$, the $2j$-th element will be removed without being queried. If the $(2j-1)$-th element is determined to be $0$, it will be removed from the list. Thus after $\mathcal{O}(\log(K))$ rounds, we have only one element left in the set, and it suffices to query this element to determine the output of \orn function.

\section{Conclusions and Future Work}
For four noisy computing tasks --- the \orn function from noisy queries, and the $\mathsf{MAX}, \mathsf{SEARCH}$, and $\mathsf{SORT}$ functions from noisy pairwise comparisons --- we  tighten the lower bounds for fixed-length noisy computing and provide  matching bounds for variable-length noisy computing. Making the bounds match exactly in the fixed-length setting  remains  an important open problem.  

\section*{Acknowledgements}
The authors would like to thank Yanjun Han for insightful discussions on the information-theoretic lower bounds. 
This work was supported in part by the NSERC Discovery Grant No. RGPIN-2019-05448, the NSERC Collaborative Research and Development Grant CRDPJ 54367619, NSF Grants IIS-1901252 and CCF-1909499.

\bibliography{ref}

\begin{thebibliography}{52}
\providecommand{\natexlab}[1]{#1}
\providecommand{\url}[1]{\texttt{#1}}
\expandafter\ifx\csname urlstyle\endcsname\relax
  \providecommand{\doi}[1]{doi: #1}\else
  \providecommand{\doi}{doi: \begingroup \urlstyle{rm}\Url}\fi

\bibitem[Agarwal et~al.(2017)Agarwal, Agarwal, Assadi, and Khanna]{Agarwal2017}
A.~Agarwal, S.~Agarwal, S.~Assadi, and S.~Khanna.
\newblock Learning with limited rounds of adaptivity: Coin tossing, multi-armed
  bandits, and ranking from pairwise comparisons.
\newblock In \emph{Proceedings of the 2017 Conference on Learning Theory},
  volume~65 of \emph{Proceedings of Machine Learning Research}, pages 39--75.
  PMLR, 07--10 Jul 2017.

\bibitem[Ailon(2011)]{Ailon2011adv}
N.~Ailon.
\newblock Active learning ranking from pairwise preferences with almost optimal
  query complexity.
\newblock In J.~Shawe-Taylor, R.~Zemel, P.~Bartlett, F.~Pereira, and K.~Q.
  Weinberger, editors, \emph{Advances in Neural Information Processing
  Systems}, volume~24. Curran Associates, Inc., 2011.

\bibitem[Ailon et~al.(2008)Ailon, Charikar, and Newman]{Ailon2008}
N.~Ailon, M.~Charikar, and A.~Newman.
\newblock Aggregating inconsistent information: Ranking and clustering.
\newblock \emph{J. ACM}, 55\penalty0 (5), nov 2008.
\newblock ISSN 0004-5411.
\newblock \doi{10.1145/1411509.1411513}.

\bibitem[Ailon et~al.(2011)Ailon, Begleiter, and Ezra]{Ailon-Begleiter2011}
N.~Ailon, R.~Begleiter, and E.~Ezra.
\newblock A new active learning scheme with applications to learning to rank
  from pairwise preferences.
\newblock 2011.

\bibitem[Ajtai et~al.(2009)Ajtai, Feldman, Hassidim, and Nelson]{Ajtai2009}
M.~Ajtai, V.~Feldman, A.~Hassidim, and J.~Nelson.
\newblock Sorting and selection with imprecise comparisons.
\newblock In \emph{ACM Trans. Algorithms}, volume~12, pages 37--48, 07 2009.
\newblock \doi{10.1007/978-3-642-02927-1_5}.

\bibitem[Audibert et~al.(2010)Audibert, Bubeck, and Munos]{audibert2010best}
J.-Y. Audibert, S.~Bubeck, and R.~Munos.
\newblock Best arm identification in multi-armed bandits.
\newblock In \emph{COLT}, pages 41--53, 2010.

\bibitem[Auer et~al.(1995)Auer, Cesa-Bianchi, Freund, and
  Schapire]{auer1995gambling}
P.~Auer, N.~Cesa-Bianchi, Y.~Freund, and R.~E. Schapire.
\newblock Gambling in a rigged casino: The adversarial multi-armed bandit
  problem.
\newblock In \emph{Proceedings of IEEE 36th annual foundations of computer
  science}, pages 322--331. IEEE, 1995.

\bibitem[Berlekamp(1964)]{Berlekamp1964}
E.~R. Berlekamp.
\newblock \emph{Block coding with noiseless feedback}.
\newblock {Ph.D.} thesis, MIT, Cambridge, MA, USA, 1964.

\bibitem[Bradley and Terry(1952)]{BTL1952}
R.~A. Bradley and M.~E. Terry.
\newblock Rank analysis of incomplete block designs: I. the method of paired
  comparisons.
\newblock \emph{Biometrika}, 39\penalty0 (3/4):\penalty0 324--345, 1952.
\newblock ISSN 00063444.

\bibitem[Braverman and Mossel(2009)]{Braverman2009}
M.~Braverman and E.~Mossel.
\newblock Sorting from noisy information.
\newblock 2009.
\newblock URL \url{https://arxiv.org/abs/0910.1191}.

\bibitem[Bretagnolle and Huber(1979)]{bretagnolle79}
J.~Bretagnolle and C.~Huber.
\newblock Estimation des densit{\'e}s: risque minimax.
\newblock \emph{Zeitschrift f{\"u}r Wahrscheinlichkeitstheorie und Verwandte
  Gebiete}, 47\penalty0 (2):\penalty0 119--137, 1979.
\newblock \doi{10.1007/BF00535278}.

\bibitem[Bubeck et~al.(2009)Bubeck, Munos, and Stoltz]{bubeck2009pure}
S.~Bubeck, R.~Munos, and G.~Stoltz.
\newblock Pure exploration in multi-armed bandits problems.
\newblock In \emph{International conference on Algorithmic learning theory},
  pages 23--37. Springer, 2009.

\bibitem[Burnashev(1976)]{Burn1976}
M.~V. Burnashev.
\newblock Data transmission over a discrete channel with feedback. random
  transmission time.
\newblock \emph{Problemy Peredachi Informatsii}, 12\penalty0 (4):\penalty0
  10--30, 1976.

\bibitem[Burnashev and Zigangirov(1974)]{Burnashev1974}
M.~V. Burnashev and K.~Zigangirov.
\newblock An interval estimation problem for controlled observations.
\newblock \emph{Problemy Peredachi Informatsii}, 10\penalty0 (3):\penalty0
  51--61, 1974.

\bibitem[Chen et~al.(2017)Chen, Chen, and Li]{Chen2017}
X.~Chen, Y.~Chen, and X.~Li.
\newblock Asymptotically optimal sequential design for rank aggregation.
\newblock \emph{Math. Oper. Res.}, 10 2017.
\newblock \doi{10.1287/moor.2021.1209}.

\bibitem[Chen and Suh(2015)]{Chen2015}
Y.~Chen and C.~Suh.
\newblock Spectral {MLE}: {Top-K} rank aggregation from pairwise comparisons.
\newblock In F.~Bach and D.~Blei, editors, \emph{Proceedings of the 32nd
  International Conference on Machine Learning}, volume~37 of \emph{Proceedings
  of Machine Learning Research}, pages 371--380, Lille, France, 07--09 Jul
  2015. PMLR.

\bibitem[Dobrushin and Ortyukov(1977{\natexlab{a}})]{dobrushin1977lower}
R.~L. Dobrushin and S.~Ortyukov.
\newblock Lower bound for the redundancy of self-correcting arrangements of
  unreliable functional elements.
\newblock \emph{Problemy Peredachi Informatsii}, 13\penalty0 (1):\penalty0
  82--89, 1977{\natexlab{a}}.

\bibitem[Dobrushin and Ortyukov(1977{\natexlab{b}})]{dobrushin1977upper}
R.~L. Dobrushin and S.~Ortyukov.
\newblock Upper bound on the redundancy of self-correcting arrangements of
  unreliable functional elements.
\newblock \emph{Problemy Peredachi Informatsii}, 13\penalty0 (3):\penalty0
  56--76, 1977{\natexlab{b}}.

\bibitem[Evans and Pippenger(1998)]{evans1998average}
W.~Evans and N.~Pippenger.
\newblock Average-case lower bounds for noisy boolean decision trees.
\newblock \emph{SIAM Journal on Computing}, 28\penalty0 (2):\penalty0 433--446,
  1998.

\bibitem[Falahatgar et~al.(2017)Falahatgar, Orlitsky, Pichapati, and
  Suresh]{Falahatgar2017}
M.~Falahatgar, A.~Orlitsky, V.~Pichapati, and A.~T. Suresh.
\newblock Maximum selection and ranking under noisy comparisons.
\newblock In D.~Precup and Y.~W. Teh, editors, \emph{Proceedings of the 34th
  International Conference on Machine Learning}, volume~70 of \emph{Proceedings
  of Machine Learning Research}, pages 1088--1096. PMLR, 06--11 Aug 2017.

\bibitem[Feige et~al.(1994)Feige, Raghavan, Peleg, and
  Upfal]{feige1994computing}
U.~Feige, P.~Raghavan, D.~Peleg, and E.~Upfal.
\newblock Computing with noisy information.
\newblock \emph{SIAM Journal on Computing}, 23\penalty0 (5):\penalty0
  1001--1018, 1994.

\bibitem[Gabillon et~al.(2012)Gabillon, Ghavamzadeh, and
  Lazaric]{gabillon2012best}
V.~Gabillon, M.~Ghavamzadeh, and A.~Lazaric.
\newblock Best arm identification: A unified approach to fixed budget and fixed
  confidence.
\newblock \emph{Advances in Neural Information Processing Systems}, 25, 2012.

\bibitem[G{\'a}cs and G{\'a}l(1994)]{gacs1994lower}
P.~G{\'a}cs and A.~G{\'a}l.
\newblock Lower bounds for the complexity of reliable boolean circuits with
  noisy gates.
\newblock \emph{IEEE Transactions on Information Theory}, 40\penalty0
  (2):\penalty0 579--583, 1994.

\bibitem[Garivier and Kaufmann(2016)]{garivier2016optimal}
A.~Garivier and E.~Kaufmann.
\newblock Optimal best arm identification with fixed confidence.
\newblock In \emph{Conference on Learning Theory}, 2016.

\bibitem[Gu and Xu(2023)]{gu2023optimal}
Y.~Gu and Y.~Xu.
\newblock Optimal bounds for noisy sorting, 2023.

\bibitem[Heckel et~al.(2019)Heckel, Shah, Ramchandran, and
  Wainwright]{Heckel2019}
R.~Heckel, N.~B. Shah, K.~Ramchandran, and M.~J. Wainwright.
\newblock {Active ranking from pairwise comparisons and when parametric
  assumptions do not help}.
\newblock \emph{Ann. Stat.}, 47\penalty0 (6):\penalty0 3099 -- 3126, 2019.

\bibitem[Horstein(1963)]{Horstein1963}
M.~Horstein.
\newblock Sequential transmission using noiseless feedback.
\newblock \emph{{IEEE} Trans. Inf. Theory}, 9\penalty0 (3):\penalty0 136--143,
  1963.
\newblock \doi{10.1109/TIT.1963.1057832}.

\bibitem[Ingster et~al.(2003)Ingster, Ingster, and
  Suslina]{ingster2003nonparametric}
Y.~Ingster, J.~I. Ingster, and I.~Suslina.
\newblock \emph{Nonparametric goodness-of-fit testing under Gaussian models},
  volume 169.
\newblock Springer Science \& Business Media, 2003.

\bibitem[Jamieson and Nowak(2014)]{jamieson2014best}
K.~Jamieson and R.~Nowak.
\newblock Best-arm identification algorithms for multi-armed bandits in the
  fixed confidence setting.
\newblock In \emph{48th Annual Conference on Information Sciences and Systems},
  pages 1--6, 2014.

\bibitem[Jamieson and Nowak(2011)]{Jamieson2011}
K.~G. Jamieson and R.~D. Nowak.
\newblock Active ranking using pairwise comparisons.
\newblock In \emph{Proceedings of the 24th International Conference on Neural
  Information Processing Systems}, NIPS'11, page 2240–2248, Red Hook, NY,
  USA, 2011. Curran Associates Inc.
\newblock ISBN 9781618395993.

\bibitem[Karp and Kleinberg(2007)]{Karp2007}
R.~M. Karp and R.~Kleinberg.
\newblock Noisy binary search and its applications.
\newblock In \emph{Proceedings of the Eighteenth Annual ACM-SIAM Symposium on
  Discrete Algorithms}, SODA '07, page 881–890, USA, 2007. Society for
  Industrial and Applied Mathematics.
\newblock ISBN 9780898716245.

\bibitem[Kaufmann et~al.(2016)Kaufmann, Capp{\'e}, and
  Garivier]{kaufmann2016complexity}
E.~Kaufmann, O.~Capp{\'e}, and A.~Garivier.
\newblock On the complexity of best arm identification in multi-armed bandit
  models.
\newblock \emph{Journal of Machine Learning Research}, 17:\penalty0 1--42,
  2016.

\bibitem[LeCam(1973)]{lecam1973convergence}
L.~LeCam.
\newblock Convergence of estimates under dimensionality restrictions.
\newblock \emph{The Annals of Statistics}, pages 38--53, 1973.

\bibitem[Mao et~al.(2018)Mao, Weed, and Rigollet]{Mao2018}
C.~Mao, J.~Weed, and P.~Rigollet.
\newblock Minimax rates and efficient algorithms for noisy sorting.
\newblock In \emph{Proceedings of Algorithmic Learning Theory}, volume~83 of
  \emph{Proceedings of Machine Learning Research}, pages 821--847. PMLR, 07--09
  Apr 2018.

\bibitem[Mitzenmacher and Upfal(2017)]{mitzenmacher2017probability}
M.~Mitzenmacher and E.~Upfal.
\newblock \emph{Probability and computing: Randomization and probabilistic
  techniques in algorithms and data analysis}.
\newblock Cambridge university press, 2017.

\bibitem[Mohajer et~al.(2017)Mohajer, Suh, and Elmahdy]{Mohajer2017}
S.~Mohajer, C.~Suh, and A.~Elmahdy.
\newblock Active learning for top-$k$ rank aggregation from noisy comparisons.
\newblock In \emph{Proceedings of the 34th International Conference on Machine
  Learning - Volume 70}, ICML'17, page 2488–2497. JMLR.org, 2017.

\bibitem[Negahban et~al.(2012)Negahban, Oh, and Shah]{Negahban2012}
S.~Negahban, S.~Oh, and D.~Shah.
\newblock Iterative ranking from pair-wise comparisons.
\newblock In \emph{Advances in Neural Information Processing Systems},
  volume~25, 2012.

\bibitem[Pelc(1989)]{pelc1989searching}
A.~Pelc.
\newblock Searching with known error probability.
\newblock \emph{Theoretical Computer Science}, 63\penalty0 (2):\penalty0
  185--202, 1989.

\bibitem[Pippenger et~al.(1991)Pippenger, Stamoulis, and
  Tsitsiklis]{pippenger1991lower}
N.~Pippenger, G.~D. Stamoulis, and J.~N. Tsitsiklis.
\newblock On a lower bound for the redundancy of reliable networks with noisy
  gates.
\newblock \emph{IEEE Transactions on Information Theory}, 37\penalty0
  (3):\penalty0 639--643, 1991.

\bibitem[Rajkumar and Agarwal(2014)]{Rajkumar2014}
A.~Rajkumar and S.~Agarwal.
\newblock A statistical convergence perspective of algorithms for rank
  aggregation from pairwise data.
\newblock In \emph{Proceedings of the 31st International Conference on
  International Conference on Machine Learning - Volume 32}, page
  I–118–I–126, 2014.

\bibitem[Reischuk and Schmeltz(1991)]{reischuk1991reliable}
R.~Reischuk and B.~Schmeltz.
\newblock Reliable computation with noisy circuits and decision trees-a general
  n log n lower bound.
\newblock In \emph{[1991] Proceedings 32nd Annual Symposium of Foundations of
  Computer Science}, pages 602--611. IEEE Computer Society, 1991.

\bibitem[Ren et~al.(2018)Ren, Liu, and Shroff]{ren2018}
W.~Ren, J.~Liu, and N.~B. Shroff.
\newblock Pac ranking from pairwise and listwise queries: Lower bounds and
  upper bounds.
\newblock 2018.

\bibitem[R{\'e}nyi(1961)]{Renyi1961}
A.~R{\'e}nyi.
\newblock On a problem of information theory.
\newblock \emph{MTA Mat. Kut. Int. Kozl. B}, 6\penalty0 (MR143666):\penalty0
  505--516, 1961.

\bibitem[Shah and Wainwright(2018)]{Shah2018}
N.~B. Shah and M.~J. Wainwright.
\newblock Simple, robust and optimal ranking from pairwise comparisons.
\newblock \emph{J. Mach. Learn. Res.}, 18\penalty0 (199):\penalty0 1--38, 2018.

\bibitem[Shah et~al.(2016)Shah, Balakrishnan, Guntuboyina, and
  Wainwright]{Shah2016}
N.~B. Shah, S.~Balakrishnan, A.~Guntuboyina, and M.~J. Wainwright.
\newblock Stochastically transitive models for pairwise comparisons:
  Statistical and computational issues.
\newblock In \emph{Proceedings of the 33rd International Conference on
  International Conference on Machine Learning - Volume 48}, ICML'16, page
  11–20. JMLR.org, 2016.

\bibitem[Tsybakov(2004)]{tsybakov2004introduction}
A.~B. Tsybakov.
\newblock Introduction to nonparametric estimation.
\newblock \emph{Springer}, 9\penalty0 (10), 2004.

\bibitem[Ulam(1976)]{Ulam1976}
S.~Ulam.
\newblock \emph{Adventures of a mathematician}.
\newblock Charles Scribner's Sons, New York, NY, USA, 1976.

\bibitem[Von~Neumann(1956)]{von1956probabilistic}
J.~Von~Neumann.
\newblock Probabilistic logics and the synthesis of reliable organisms from
  unreliable components.
\newblock \emph{Automata studies}, 34:\penalty0 43--98, 1956.

\bibitem[Wang et~al.(2022)Wang, Ghaddar, and Wang]{wang2022noisy}
Z.~Wang, N.~Ghaddar, and L.~Wang.
\newblock Noisy sorting capacity.
\newblock In \emph{2022 IEEE International Symposium on Information Theory
  (ISIT)}, pages 2541--2546. IEEE, 2022.

\bibitem[Wang et~al.(2023)Wang, Ghaddar, Zhu, and Wang]{isit2022paper}
Z.~Wang, N.~Ghaddar, B.~Zhu, and L.~Wang.
\newblock Noisy sorting capacity, 2023.
\newblock URL \url{https://arxiv.org/abs/2202.01446}.

\bibitem[Wauthier et~al.(2013)Wauthier, Jordan, and Jojic]{Wauthier2013}
F.~Wauthier, M.~Jordan, and N.~Jojic.
\newblock Efficient ranking from pairwise comparisons.
\newblock In \emph{Proceedings of the 30th International Conference on Machine
  Learning}, volume~28 of \emph{Proceedings of Machine Learning Research},
  pages 109--117, Atlanta, Georgia, USA, 17--19 Jun 2013.

\bibitem[Yu(1997)]{yu1997assouad}
B.~Yu.
\newblock Assouad, fano, and le cam.
\newblock In \emph{Festschrift for Lucien Le Cam}, pages 423--435. Springer,
  1997.

\end{thebibliography}

\appendices

\section{Useful Lemmas}\label{app:lemma}
Here, we introduce several important lemmas.

\begin{lemma}[Le Cam's Two Point Lemma, see e.g.~\citep{yu1997assouad}]\label{lem:lecam}
    For any $\theta_0,\theta_1\in\Theta$, suppose that the following separation condition holds for some loss function $L(\theta, a):\Theta\times \mathcal{A}\rightarrow \mathbb{R}$:
    \begin{align*}
    \forall {a\in\mathcal{A}}, L(\theta_0, a)+L(\theta_1, a)\geq \Delta>0.
    \end{align*}
    Then we have
    \begin{align*}
        \inf_f \sup_\theta \mathbb{E}_\theta[L(\theta, f(X))]\geq \frac{\Delta}{2}\left(1-\mathsf{TV}(\mathbb{P}_{\theta_0},\mathbb{P}_{\theta_1})\right). 
    \end{align*}
    
\end{lemma}

\begin{lemma}[Point vs Mixture, see e.g.~\citep{ingster2003nonparametric}]\label{lem:mixture}
     For any $\theta_0\in\Theta$, $\Theta_1\subset \Theta$, suppose that the following separation condition holds for some loss function $L(\theta, a):\Theta\times \mathcal{A}\rightarrow \mathbb{R}$:
    \begin{align*}
    \forall \theta_1\in \Theta_1, {a\in\mathcal{A}}, L(\theta_0, a)+L(\theta_1, a)\geq \Delta>0.
    \end{align*}
    Then for any probability measure $\mu$ supported on $\Theta_1$, we have
    \begin{align*}
        \inf_f \sup_\theta \mathbb{E}_\theta[L(\theta, f(X))]\geq \frac{\Delta}{2}\left(1-\mathsf{TV}(\mathbb{P}_{\theta_0},\mathbb{E}_{\mu(d\theta)}[\mathbb{P}_{\theta_1}])\right). 
    \end{align*}
\end{lemma}
\begin{lemma}[Divergence Decomposition \citep{auer1995gambling}]\label{lem:div}
Let $T_i$ be the  random variable denoting the number of times experiment $i\in[K]$ is performed under some policy $\pi$, then for two distributions $\mathbb{P}^\pi, \mathbb{Q}^\pi$ under policy $\pi$,
\begin{align*}
    D_{\mathsf{KL}}(\mathbb{P}^\pi, \mathbb{Q}^\pi) = \sum_{i\in[K]}\mathbb{E}_{\mathbb{P}^\pi}[T_i] D_{\mathsf{KL}}(\mathbb{P}_i^\pi, \mathbb{Q}_i^\pi).
\end{align*}
\end{lemma}

\begin{lemma}[An upper Bound of Bretagnolle–Huber inequality
 (\citep{bretagnolle79}, and Lemma 2.6 in \citet{tsybakov2004introduction})]\label{lem:bh}
    For any distribution $\mathbb{P}_1,\mathbb{P}_2$, one has
    \begin{align*}
        \mathsf{TV}(\mathbb{P}_1, \mathbb{P}_2)\leq 1-\frac{1}{2} \exp(-D_{\mathsf{KL}}(\mathbb{P}_1, \mathbb{P}_2)).
    \end{align*}
\end{lemma}

\begin{lemma}
\label{lem:chisquare}(Chi-Square divergence for mixture of distributions, see e.g.~\citep{ingster2003nonparametric})
Let $(P_\theta)_{\theta\in\Theta}$ be  a family of distributions parametrized by $\theta$, and $Q$ be any fixed distribution. Then for any probability measure $\mu$ supported on $\Theta$, we have
\begin{align*}
    \chi^2(\mathbb{E}_{\theta\sim \mu}[P_\theta], Q) = \mathbb{E}_{\theta\sim\mu, \theta'\sim\mu}\left[\sum_x \frac{P_\theta(x) P_{\theta'}(x)}{Q(x)}\right]-1.
\end{align*}
Here $\theta,\theta'$ in the expectation are independent. 
\end{lemma}

\begin{lemma}\label{lem:chernoff}[Chernoff's bound for Binomial random variables {\cite[Exercise 4.7]{mitzenmacher2017probability}}]
If $X\sim \mathsf{Bin}(n, \frac{\lambda} {n})$,
then for any $\eta\in(0, 1)$, we have
\begin{align}
    \mathbb{P}(X\geq (1+\eta)\lambda) &\leq \left(\frac{e^\eta}{(1+\eta)^{(1+\eta)}}\right)^\lambda \label{eq.chernoff_upper}\\
    \mathbb{P}(X\leq (1-\eta)\lambda) & \leq \left(\frac{e^{-\eta}}{(1-\eta)^{(1-\eta)}}\right)^\lambda\leq e^{-\eta^2\lambda/2}. \label{eq.chernoff_lower}
\end{align}
\end{lemma} 

\section{Proof for Theorem~\ref{thm:max_adaptive_lower}}
\label{proof:max_adaptive}

\begin{proof}
We first select an arbitrary sequence $0<X_1<X_2<\cdots<X_K<1$. Let $\theta_0=(X_1,X_2,X_3,\cdots,X_K)$ be original sequence which has its largest value in the $K$-th element. Thus $\mathsf{MAX}(\theta_0)=K$. Now for any $i\in[K-1]$, we design $\theta_i = ( X_1,\cdots, X_{i-1}, X_K, X_i, X_{i+1}, \cdots, X_{K-1})$, i.e., we move the $K$-th element in $\theta_0$ and insert it between $X_{i-1}$ and $X_i$. Let $T_{i,j}$ be the random variable that represents the number of comparisons between the $i$-th item and $j$-th item in the $T$ rounds. 
We know that $\mathsf{MAX}(\theta_i)=i$ for all $i\in[K-1]$.
Following a similar proof as Theorem~\ref{thm:or_adaptive_lower}, we know that
\begin{align*}
      &\inf_{\hat\mu} \sup_{\theta\in\{0, 1\}^K} \mathbb{P}(\hat\mu \neq \mathsf{MAX}(\theta)) \\
      & \geq \sup_{1\leq j\leq K-1}\frac{1}{2}(1-\mathsf{TV}(\mathbb{P}_{\theta_0}, \mathbb{P}_{\theta_j}))  \\ 
      & \geq \sup_{1\leq j\leq K-1} \frac{1}{4}\exp(-D_{\mathsf{KL}}(\mathbb{P}_{\theta_0}, \mathbb{P}_{\theta_j})) \\
      & \geq \sup_{1\leq j\leq K-1} \frac{1}{4}\exp(-\sum_{l=j+1}^K \mathbb{E}_{\theta_0}[T_{j,l}] \cdot D_{\mathsf{KL}}(p\|1-p)).
\end{align*}Now since $\sum_{i,j\in[K], i< j} \mathbb{E}_{\theta_0}[T_{i,j}] = T$, there must exists some $j$ such that $ \sum_{l=j+1}^K \mathbb{E}_{\theta_0}[T_{j,l}]  \leq T/K$. This gives
\begin{align*}
      \inf_{\hat\mu} \sup_{\theta\in\{0, 1\}^K} \mathbb{P}(\hat\mu \neq \mathsf{MAX}(\theta)) &  \geq  \frac{1}{4}\cdot \exp\left(-\frac{T\cdot D_{\mathsf{KL}}(p\|1-p)}{K }\right).
\end{align*}

 On the other hand, $K$ is naturally a lower bound for the query complexity since one has to query each element at least once. Thus we arrive at a lower bound of $\Omega(\max(K,K\log(1/\delta)/D_{\mathsf{KL}}(p\|1-p)))$. Note that this is equivalent to $\Omega(K/(1-H(p))+K\log(1/\delta)/D_{\mathsf{KL}}(p\|1-p))$ up to a constant factor when $\delta<0.49$. The reason is that when $p$ is bounded away from $0$,  $(1-H(p))/D_{\mathsf{KL}}(p\|1-p)$ is always some constant. When $p$ is close to $0$, $1-H(p)$ is within constant factor of $1$. 
\end{proof}
\section{Proof for Theorem~\ref{thm:max_non-adaptive_lower}}
\label{proof:max_non-adaptive}

\begin{proof}
Consider an arbitrary sequence $0<X_1<X_2<\cdots<X_K<1$. 
Since $\sum_{i,j\in[K], i< j} \mathbb{E}_{\theta_0}[T_{i,j}] = T$, there must exists some pair $(i,j)$ such that $ \mathbb{E}_{\theta_0}[T_{i,j}]  \leq 2T/K(K-1)$.  
 Now we construct $$\theta_0=(X_1,X_2,X_3,\cdots,X_K, \cdots, X_{K-1},\cdots, X_{K-2}),$$ where $X_K$ lies in the $i$-th position and $X_{K-1}$ lies in the $j$-th position, and $$\theta_1=(X_1,X_2,X_3,\cdots,X_{K-1}, \cdots, X_{K},\cdots, X_{K-2}),$$ where $X_K$ lies in the $j$-th position and $X_{K-1}$ lies in the $i$-th position.   Thus $\mathsf{MAX}(\theta_0)=i$, $\mathsf{MAX}(\theta_1)=j$. 
From Le Cam's two point lemma, we know that
\begin{align*}&  \inf_{\hat\mu} \sup_{\theta\in\{0, 1\}^K} \mathbb{P}(\hat\mu \neq \mathsf{MAX}(\theta)) 
    \\  & \geq \frac{1}{2}(1-\mathsf{TV}(\mathbb{P}_{\theta_0}, \mathbb{P}_{\theta_1})) \\
      & \geq \frac{1}{4}\exp(-\mathbb{E}_{\theta_0}[T_{i,j}] \cdot D_{\mathsf{KL}}(p\|1-p))
     \\
     &   \geq  \frac{1}{4}\cdot \exp\left(-\frac{2T\cdot D_{\mathsf{KL}}(p\|1-p)}{K^2}\right).
\end{align*}
 On the other hand, $K^2$ is naturally a lower bound for the query complexity since one has to query each element at least once. Thus we arrive at a lower bound of $\Omega(\max(K^2,K^2\log(1/\delta)/D_{\mathsf{KL}}(p\|1-p)))$. Note that this is equivalent to $\Omega(K^2/(1-H(p))+K^2\log(1/\delta)/D_{\mathsf{KL}}(p\|1-p))$ up to a constant factor when $\delta<0.49$.
\end{proof}

\section{Proof for Theorem~\ref{thm:search_adaptive_lower}}
\label{proof:search_adaptive}
\begin{proof}
We begin with the first half, i.e.
\begin{align*}
 \inf_{\hat\mu} \sup_{X} \mathbb{P}(\hat\mu \neq \mathsf{SEARCH}(X)) 
    \geq   \frac{1}{4}\cdot \exp\left(-{T\cdot D_{\mathsf{KL}}(p\|1-p)}\right).
\end{align*} 
To see this, simply consider the case of $K=1$, and we need to determine whether $X<X_1$ or $X>X_1$ from their pairwise comparisons. We consider two instances $X^{(0)}, X^{(1)}$, where $X^{(0)}<X_1<X^{(1)}$. From Le Cam's lemma, we have
\begin{align*}
 \inf_{\hat\mu} \sup_{X} \mathbb{P}(\hat\mu \neq \mathsf{SEARCH}(\theta))& \geq \frac{1}{2}(1-\mathsf{TV}(\mathbb{P}_{X^{(0)}}, \mathbb{P}_{X^{(1)}})) \\
      & \geq \frac{1}{4}\exp(-T \cdot D_{\mathsf{KL}}(p\|1-p)).
\end{align*}
Next, we  aim to prove the second half, namely
\begin{align*}
 & \inf_{\hat\mu} \sup_{X} \mathbb{P}(\hat\mu \neq \mathsf{SEARCH}(X)) \\
    \geq & \frac{1}{2}\cdot \left(1-\frac{T\cdot (1-H(p))+\log(2)}{\log(K)}\right).
\end{align*} 
Now we design $K$ instances $X^{(0)},\cdots, X^{(K-1)}$.  We let $X^{(0)}<X_1$, and $X^{(l)}\in(X_l, X_{l+1})$. 
From Le Cam's lemma, we have 
\begin{align*}
     &  \inf_{\hat\mu} \sup_{X} \mathbb{P}(\hat\mu \neq \mathsf{SEARCH}(\theta))\\
      & \geq  \inf_{\hat\mu} \sup_{X\in\{X^l\}_{l\in[K]}} \mathbb{P}(\hat\mu \neq \mathsf{SEARCH}(\theta)) \\ 
        & \geq  \inf_{\Psi} \frac{1}{2K}\sum_{l\in[K]} \mathbb{P}_{X^{(l)}}(\Psi \neq l).
        \end{align*}
Now by Fano's inequality, we have
\begin{align*}
      & \inf_{\Psi} \frac{1}{2K}\sum_{l\in[K]} \mathbb{P}_{^{(l)}}(\Psi \neq l)\\
        & \geq  \frac{1}{2}\cdot \left(1-\frac{I(V;Y)+\log(2)}{\log(K)}\right).
\end{align*}
Here $V\sim \mathsf{Unif}(\{0, 1,\cdots,K-1\}))$ and $Y$ satisfies $\mathbb{P}_{Y|V=l}= \mathbb{P}_{X^{(l)}}$. Following the same argument as the proof of Theorem 3 in~\citet{wang2022noisy}, we know that  $I(V;Y)\leq T\cdot (1-H(p))$. Thus overall we have
\begin{align*}
 \inf_{\hat\mu} \sup_{X} \mathbb{P}(\hat\mu \neq \mathsf{SEARCH}(X)) 
    \geq \frac{1}{4}\cdot \left(1-\frac{T\cdot (1-H(p))}{\log(K)}\right).
\end{align*}

Thus the queries required to recover the true value with probability at least  $1-\delta$ is lower bounded by $\Omega(\log(K)/(1-H(p))+\log(1/\delta)/D_{\mathsf{KL}}(p\|1-p))$.     
\end{proof}
\section{Proof for Theorem~\ref{thm:search_non-adaptive_lower}}
\label{proof:search_non-adaptive}

\begin{proof}
Let $T_i$ be the random variable that denotes the number of times $X$ is compared with $X_i$. 
Since $\sum_{i\in[K]} \mathbb{E}[T_{i}] = T$, there must exists some $i$ such that $ \mathbb{E}[T_{i}]  \leq T/K$.  
 Now we construct the first instance $X^{(0)}\in(X_{i-1},X_i)$,  and the second instance $X^{(1)}\in(X_i, X_{i+1})$.  
From Le Cam's two point lemma, we know that
\begin{align*}
& \inf_{\hat\mu} \sup_{X} \mathbb{P}(\hat\mu \neq \mathsf{SEARCH}(X))
     \\ & \geq \frac{1}{2}(1-\mathsf{TV}(\mathbb{P}_{X^{(0)}}, \mathbb{P}_{X^{(1)}})) \\
      & \geq \frac{1}{4}\exp(-\mathbb{E}[T_{i}] \cdot D_{\mathsf{KL}}(p\|1-p))
     \\
     &   \geq  \frac{1}{4}\cdot \exp\left(-\frac{T\cdot D_{\mathsf{KL}}(p\|1-p)}{K}\right).
\end{align*}

 On the other hand, $K$ is naturally a lower bound for query complexity since one has to query each element at least once. Thus we arrive at a lower bound of $\Omega(\max(K,K\log(1/\delta)/D_{\mathsf{KL}}(p\|1-p)))$.
\end{proof}
\section{Proof for Theorem~\ref{thm:sort_adaptive_lower}}\label{proof:sort_adaptive}
\begin{proof}
From~\citep{isit2022paper}, it is proven with Fano's inequality that
\begin{align*}
    \inf_{\hat\mu} \sup_{\theta\in[0, 1]^K} \mathbb{P}(\hat\mu \neq \mathsf{SORT}(\theta))\geq  \frac{1}{2}\cdot \left(1-\frac{T\cdot (1-H(p))}{K\log(K)}\right).
\end{align*}
So it suffices to prove that \begin{align*}
    \inf_{\hat\mu} \sup_{\theta\in[0, 1]^K} \mathbb{P}(\hat\mu \neq \mathsf{SORT}(\theta))\geq  \frac{1}{4}\cdot \exp\left(-\frac{T\cdot D_{\mathsf{KL}}(p\|1-p)}{K}\right).
\end{align*}
To see this, consider an arbitrary sequence $0<X_1<X_2<\cdots<X_K<1$. Now we design $K$ instances $\theta_0,\cdots, \theta_{K-1}$.  We let $\theta_0 = (X_{1},\cdots, X_{K})$, and $\theta_l$ be the instance that switches the element $X_{l}$ with $X_{l+1}$, where $l\in[K]$. 
From Le Cam's two point lemma, we have  \begin{align*}&  \inf_{\hat\mu} \sup_{\theta\in[0, 1]^K} \mathbb{P}(\hat\mu \neq \mathsf{SORT}(\theta)) \\
  & \geq \sup_{1\leq j\leq K-1}\frac{1}{2}(1-\mathsf{TV}(\mathbb{P}_{\theta_0}, \mathbb{P}_{\theta_j}))  \\ 
      & \geq \sup_{1\leq j\leq K-1} \frac{1}{4}\exp(-D_{\mathsf{KL}}(\mathbb{P}_{\theta_0}, \mathbb{P}_{\theta_j})) \\
      & \geq \sup_{1\leq j\leq K-1} \frac{1}{4}\exp(-\mathbb{E}_{\theta_0}[T_{j,j+1}] \cdot D_{\mathsf{KL}}(p\|1-p)).
\end{align*}
Let $T_{i,j}$ be the random variable that represents the number of comparisons between the $i$-th item and $j$-th item in the $T$ rounds.  Now since $\sum_{1\leq j\leq K-1} \mathbb{E}_{\theta_0}[T_{j,j+1}] \leq T$, there must exists some $j$ such that $ \mathbb{E}_{\theta_0}[T_{j,j+1}]  \leq T/K$. This gives
\begin{align*}
      \inf_{\hat\mu} \sup_{\theta\in[0, 1]^K} \mathbb{P}(\hat\mu \neq \mathsf{SORT}(\theta))&  \geq  \frac{1}{4}\cdot \exp\left(-\frac{T\cdot D_{\mathsf{KL}}(p\|1-p)}{K }\right).
\end{align*}
Altogether, this shows a lower bound on the query complexity $\Omega(K\log(K)/(1-H(p))+K\log(1/\delta)/D_{\mathsf{KL}}(p\|1-p)))$. Note that this is equivalent to $\Omega(K\log(K)/(1-H(p))+K\log(K/\delta)/D_{\mathsf{KL}}(p\|1-p)))$ since $1/(1-H(p))\gtrsim 1/D_{\mathsf{KL}}(p\|1-p)$.
\end{proof}

\section{Proof for Theorem~\ref{thm:sort_non-adaptive_lower}}
\label{proof:sort_non-adaptive}

\begin{proof}
We first select an arbitrary sequence $0<X_1<X_2<\cdots<X_K<1$. Let $\sigma_i:[K]\mapsto [K]$  be the $i$-th permutation of the sequence, where $i\in[K!]$. Here $\sigma_i(k)$ refers to the $k$-th largest element under permutation $\sigma_i$. Now we consider a summation $  \sum_{i\in[K!],j\in[K-1]} \mathbb{E}[T_{\sigma_i(j), \sigma_i(j+1)}]$. For each pair $(i,j)$, $\mathbb{E}[T_{i,j}]$ is counted $2(K-1)!$ times in the summation. Furthermore, we know that $\sum_{i<j} \mathbb{E}[T_{i,j}] = T$. Thus we have  
\begin{align*}
    \sum_{i\in[K!],j\in[K-1]} \mathbb{E}[T_{\sigma_i(j), \sigma_i(j+1)}] = 2T(K-1)!
\end{align*}
We know that there must exists some $i$ such that 
\begin{align*}
    \sum_{j\in[K]} \mathbb{E}[T_{\sigma_i(j), \sigma_i(j+1)}]\leq \frac{2T(K-1)!}{K!}= \frac{2T}{K}.
\end{align*}
Now we design $K$ instances $\theta_0,\cdots, \theta_{K-1}$.  We let $\theta_0 = ( X_{\sigma_i(1)},\cdots, X_{\sigma_i(K)})$, and $\theta_l$ be the instance that switches the element $X_{\sigma_i(l)}$ with $X_{\sigma_i(l+1)}$ {based on $\theta_0$}, where {$l\in[K-1]$}. 
From Le Cam's two point lemma, we have 
\begin{align*}
     &  \inf_{\hat\mu} \sup_{\theta\in\{0, 1\}^K} \mathbb{P}(\hat\mu \neq \mathsf{SORT}(\theta))\\
      & \geq  \inf_{\hat\mu} \sup_{\theta\in\{\theta_l\}_{l\in[K]}} \mathbb{P}(\hat\mu \neq \mathsf{SORT}(\theta)) \\ 
        & \geq  \inf_{\Psi} \frac{1}{2K}\sum_{l\in[K]} \mathbb{P}_{\theta_l}(\Psi \neq l).
        \end{align*}
Now by Fano's inequality, we have
\begin{align*}
      & \inf_{\Psi} \frac{1}{2K}\sum_{l\in[K]} \mathbb{P}_{\theta_l}(\Psi \neq l)\\
        & \geq  \frac{1}{2}\cdot \left(1-\frac{I(V;X)+\log(2)}{\log(K)}\right) \\
      & =  \frac{1}{2}\cdot \left(1-\frac{\sum_{l\in[K]}D_{\mathsf{KL}}(\mathbb{P}^l, \bar{\mathbb{P}})/K+\log(2)}{\log(K)}\right).
\end{align*}
Here $\bar{\mathbb{P}}=\frac{1}{K} \sum_{l\in[K]}\mathbb{P}^l$. 
We can compute
\begin{align*}
   & D_{\mathsf{KL}}(\mathbb{P}^l, \bar{\mathbb{P}})\\ &\leq \mathbb{E}[T_{\sigma_i(l),\sigma_i(l+1)}] D_{\mathsf{KL}}\left(1-p\|\frac{(K-1)p+(1-p)}{K}\right) \\
    &  +\sum_{m\neq l}\mathbb{E}[T_{\sigma_i(m),\sigma_i(m+1)}]D_{\mathsf{KL}}\left(p\| \frac{(K-1)p+(1-p)}{K}\right) \\ 
    & \leq \frac{K-1}{K}\cdot \mathbb{E}[T_{\sigma_i(l),\sigma_i(l+1)}] D_{\mathsf{KL}}\left(p\|1-p\right) \\
    &  +\frac{1}{K}\sum_{m\neq l}\mathbb{E}[T_{\sigma_i(m),\sigma_i(m+1)}]D_{\mathsf{KL}}\left(p\| 1-p\right).
\end{align*}
Now summing over all $l$, we have 
\begin{align*}
    \frac{1}{K}\sum_{l\in[K]}D_{\mathsf{KL}}(\mathbb{P}^l, \bar{\mathbb{P}})& \leq \frac{2D_{\mathsf{KL}}\left(p\| 1-p\right)}{K} \cdot   \sum_{l\in[K]} \mathbb{E}[T_{\sigma_i(l), \sigma_i(l+1)}] \\ 
    & \leq  \frac{4TD_{\mathsf{KL}}\left(p\| 1-p\right)}{K^2}.
\end{align*}
  This gives
\begin{align*}
      \inf_{\hat\mu} \sup_{\theta\in\{0, 1\}^K} \mathbb{P}(\hat\mu \neq \mathsf{OR}(\theta)) &  \geq  \frac{1}{2}\cdot \left(1-\frac{4T\cdot D_{\mathsf{KL}}(p\|1-p)}{K^2\log(K)}\right).
\end{align*}

This shows that to output the correct answer with probability at least $2/3$, one needs at least $C \cdot K^2\log(K)/D_{\mathsf{KL}}(p\|1-p)$ queries for some universal constant $C$.

\end{proof}

\section{Upper Bounds for Non-adaptive Sampling} \label{app:non-adaptive_upper}

In this section, we present a theorem on the upper bounds for non-adaptive learning in the worst-case query model.
\begin{theorem}\label{thm:non-adaptive_upper}
    One can design algorithm such that the worst-case query complexity is 
    \begin{enumerate}
        \item $\mathcal{O}(\frac{K\log(K/\delta)}{1-H(p)})$ for  computing $\mathsf{OR}$;
        \item $\mathcal{O}(\frac{K^2\log(K/\delta)}{1-H(p)})$ for  computing $\mathsf{MAX}$;
        \item $\mathcal{O}(\frac{K\log(1/\delta)}{1-H(p)})$ for  computing $\mathsf{SEARCH}$;
        \item $\mathcal{O}(\frac{K^2\log(K/\delta)}{1-H(p)})$ for  computing $\mathsf{SORT}$;
        
    \end{enumerate}
\end{theorem}
The results for  $\mathsf{OR}$,  $\mathsf{MAX}$ and $\mathsf{SORT}$ are based the simple algorithms of querying all possible elements equal number of times. And the analysis is a direct union bound argument. Here we only present  the algorithm and analysis for  $\mathsf{SEARCH}$.
\begin{proof}
    Assume that the target $X$ lies between $X_l$ and $X_{l+1}$. 
Consider the non-adaptive learning algorithm which compares $X$ with each element  $X_i$ for $T_i = \lfloor T/K \rfloor  = 4\log(1/\delta)/(1-H(p))$ times. 

Let $N_i$ be the number of observations of $1$ among $T_i$ queries for element $X_i$. Consider the following algorithm:
\begin{align*}
    \hat l = \argmax_{l\in[K]} \sum_{i=1}^l N_i +\sum_{i=l+1}^K (T_i-N_i). 
\end{align*}
We show that with high probability, $\hat l = l$. We have
\begin{align*}
    \mathbb{P}(\hat l \neq l) & \leq \sum_{j\neq l} \mathbb{P}(\hat l = j)\\
    & \leq \sum_{j\neq l} \mathbb{P}\Big(\sum_{i=1}^l N_i +\sum_{i=l+1}^K (T_i-N_i)   - (\sum_{i=1}^j N_i +\sum_{i=j+1}^K (T_i-N_i))<0\Big).
\end{align*}
Now we bound the above probability for each $j$. Without loss of generality, we assume that $j>l$. The above probability can be written as
\begin{align*}
& \mathbb{P}\Big(\sum_{i=1}^l N_i +\sum_{i=l+1}^K (T_i-N_i) < \sum_{i=1}^j N_i +\sum_{i=j+1}^K (T_i-N_i)\Big) \\
 =    & \mathbb{P}\Big(\sum_{i=l+1}^j (T_i-2N_i)<0\Big) \\
 = &  \mathbb{P}\Big(\sum_{i=l+1}^j N_i  > \frac{1}{2} (j-l) \lfloor T/K \rfloor \Big) 
\end{align*}
Note that $\sum_{i=l+1}^j N_i \sim \mathsf{Bin}((j-l) \lfloor T/K \rfloor ,  p)$. Let $n = \frac{1}{2} (j-l) \lfloor T/K \rfloor $. From Lemma~\ref{lem:chernoff} we have
\begin{align*}
  \mathbb{P}(\sum_{i=l+1}^j N_i\geq n/2)  &\leq \left(\frac{e^\frac{1-2p}{2p}}{(1/2p)^{(1/2p)}}\right)^{np}  = ( 2p\cdot \exp(1-2p))^{n/2} \\ 
    &< \exp\left(\log(1/\delta)\cdot {\frac{2(j-l)(1-2p+\log(2p))}{(1-2p)^2}}\right)   < \delta^{j-l}.
\end{align*}
Now by summing over the probability for different $j's$, we get
\begin{align*}
    \mathbb{P}(\hat l \neq l) \leq \sum_{j\neq l} \delta^{|j-l|
}< \frac{2\delta}{1-\delta}.
\end{align*}
Rescaling $\delta$ finishes the proof. 
\end{proof}

\section{Upper Bounds for Adaptive Sampling}\label{app:upper_or}
Here we present the tournament algorithm for computing \orn\! introduced in~\citep{feige1994computing}. Similar algorithm can also be applied to compute \maxn\!.  The main difference is that in $\mathsf{MAX}$, we directly compare two elements $\lceil\frac{4(2i-1)\log(1/\delta)}{(1-H(p))}\rceil$ times instead of comparing their number of $1'$s.

\begin{algorithm}[!htbp]
\caption{Tournament for computing \orn with noise}
\label{alg:tournament}
  \begin{algorithmic}[1]
\State  \textbf{Input}: Target confidence level $\delta$.
  \State Set $\mathcal{X} = (X_1,X_2,\cdots, X_K)$ as the list of all bits with unknown value.
\For{iteration $i= 1:\lceil\log_2(K)\rceil $}
\State Query  $\lceil\frac{4(2i-1)\log(1/\delta)}{(1-H(p))}\rceil$ times each of the element in $\mathcal{X}$.
 \For{iteration $j = 1:\lceil |\mathcal{X}|/2\rceil $}
     \State Compare the number of $1$'s in the queries from the $(2j-1)$-th element and $(2j)$-th element, remove the element with smaller number of $1$'s from the list $\mathcal{X}$. 
     Ties are broken arbitrarily. (If the $(2j)$-th element does not exist, we will not remove the $(2j-1)$-th element.)
      \EndFor
      \State Break when $\mathcal{X}$ only has one element left. 
      \EndFor
\State Query  $\lceil \frac{6\log(1/\delta)}{(1-H(p))}\rceil $ times the only left element in $\mathcal{X}$. Return $1$ if there is more than half $1$'s, and $0$ otherwise.
\end{algorithmic}
  \end{algorithm}
  The following theorem is due to~\citep{feige1994computing}. We include it here for completeness.
\begin{theorem}\label{thm:or_upper}
Algorithm~\ref{alg:tournament} finishes within $C\cdot K\log(1/\delta)/(1-H(p))$ queries, and outputs the correct value of $\mathsf{OR}(X_1,\cdots,X_K)$ with probability at least $1-2\delta$ when $\delta<1/2$.
\end{theorem}

\begin{proof}
First, we compute the total number of queries of the algorithm. Without loss of generality, we may assume that $K$ can be written as $2^m$ for some integer $m$. If not we may add no more than $K$ extra dummy $0$'s to the original list $\mathcal{X}$ to make sure $K=2^m$. In each of the outer iteration $i$, the size of $\mathcal{X}$ is decreased half. We know that after $\lceil \log_2(K)\rceil$ round, the set $\mathcal{X}$ will only contain one element. In round $i$, the number of queries we make for each element is $\lceil\frac{4(2i-1)\log(1/\delta)}{(1-2p)^2}\rceil$. The total number of queries we make is
\begin{align*}
   & \left \lceil\frac{4\log(1/\delta)}{(1-2p)^2}\right\rceil + \sum_{i=1}^{\lceil \log_2(K)\rceil}\left\lceil\frac{4(2i-1)\log(1/\delta)}{(1-2p)^2}\right\rceil \cdot \frac{K}{2^{i-1}} \\
  \leq   &\left \lceil\frac{4\log(1/\delta)}{(1-2p)^2}\right\rceil + \sum_{i=1}^{\lceil \log_2(K)\rceil}\left(\frac{4(2i-1)\log(1/\delta)}{(1-2p)^2}+1\right) \cdot \frac{K}{2^{i-1}} \\
    \leq & \left \lceil\frac{4\log(1/\delta)}{(1-2p)^2}\right\rceil + 2K + \frac{K\log(1/\delta)}{(1-2p)^2} \sum_{i=1}^{\lceil \log_2(K)\rceil}\frac{4(2i-1)}{2^{i-1}} \\
   \leq &  \left \lceil\frac{4\log(1/\delta)}{(1-2p)^2}\right\rceil + 2K + \frac{28 K\log(1/\delta)}{(1-2p)^2} \\ 
   \leq &  \frac{ C K\log(1/\delta)}{1-H(p)} 
\end{align*}
Here in last inequality we use the fact below, which can be verified numerically: 
\begin{align}
\forall p\in[0, 1],    \frac{(1/2-p)^2}{1-H(p)} \in [1/4,1/2]. 
\end{align}

Now we show that the failure probability of the algorithm is at most $\delta$. Consider the first case where all the elements are $0$. Then no matter which element is left in $\mathcal{X}$, the probability that the algorithm fails is the probability that a Binomial random variable $X\sim B(n, np)$ has value larger or equal to $n/2$ with $n=\left\lceil\frac{4\log(1/\delta)}{(1-2p)^2}\right\rceil$.


Taking $\lambda = np$, $\eta = \frac{1-2p}{2p}$ in Equation (\ref{eq.chernoff_upper}) of Lemma~\ref{lem:chernoff}, we know that
\begin{align*}
    \mathbb{P}(X\geq n/2) &\leq \left(\frac{e^\frac{1-2p}{2p}}{(1/2p)^{(1/2p)}}\right)^{np} \\ 
    & = ( 2p\cdot \exp(1-2p))^{n/2} \\ 
    &< \exp\left(\log(1/\delta)\cdot {\frac{2(1-2p+\log(2p))}{(1-2p)^2}}\right) \\
    & < \delta.
\end{align*}
Here the last inequality uses the fact that ${\frac{(1-2p+\log(2p))}{(1-2p)^2}}<-1/2$ for all $p\in[0,1/2)$. This shows that the final failure probability is bounded by $\delta$ when the true elements are all $0$.

Consider the second case where there exists at least a $1$ in the original elements $X_1,\cdots, X_K$. Without loss of generality, we assume that $X_1=1$. Let $\mathcal{X}^i$ be the remaining list of elements at the beginning of $i$-th iteration. We let $\mathcal{E}_i$ be the event that the first element in $\mathcal{X}^i$ is $1$ while the first element in $\mathcal{X}^{i+1}$ is $0$. This event only happens when the second element in $\mathcal{X}^i$ is $0$ and gets more $1$'s  in the noisy queries than the first element. Let $\mathcal{A}$ denote the event that the only left element is $1$ in the last round,  we have 
\begin{align*}
    \mathbb{P}(\mathcal{A}) & \geq  1 - \mathbb{P}\left(\bigcup_{i=1}^{\lceil \log_2(K)\rceil} \mathcal{E}_i\right)     \\
    & \geq 1-\sum_{i=1}^{\lceil \log_2(K)\rceil} \mathbb{P}(\mathcal{E}_i) \\ 
    & \geq  1-\sum_{i=1}^{\lceil \log_2(K)\rceil} \mathbb{P}(Y_i - X_i\geq 0),
\end{align*}
where $X_i\sim B(n_i, n_i(1-p))$ and $Y_i\sim B(n_i, n_ip)$, with $n_i = \lceil\frac{4(2i-1)\log(1/\delta)}{(1-2p)^2}\rceil$. Let $Z_i = Y_i-X_i + n_i$, we know that  the random variable $Z_i\sim B(2n_i, n_ip)$. Thus we have
\begin{align*}
      \mathbb{P}(\mathcal{A}) & \geq 1-\sum_{i=1}^{\lceil \log_2(K)\rceil} \mathbb{P}(Z_i\geq n_i) \\ 
      & >1-\sum_{i=1}^{\lceil \log_2(K)\rceil} \delta^{2(2i-1)} \\ 
      & > 1- \frac{\delta^2}{1-\delta^2}.
\end{align*}
Following the same argument on Binomial distribution, we can upper bound the error probability under event $\mathcal{A}$ with $\delta$. Thus the total failure probability is upper bounded by $\delta + \delta^2/(1-\delta^2)<2\delta$ when $\delta<1/2$.

\end{proof}
\section{Proof of Theorem~\ref{thm:variable_length}}\label{proof:variable_length}

\subsection{Lower Bounds}
First, note that all our  lower bounds for fixed-length can be adapted to variable-length by replacing $T$ with $\mathbb{E}[T]$. The bound of mutual information in Fano's inequality in Section~\ref{proof:search_adaptive} can be proven using the same argument as Lemma 27 in~\citet{gu2023optimal}, and the divergence decomposition lemma (Lemma~\ref{lem:div}) still holds for variable length due to Lemma 15 in \citet{kaufmann2016complexity}.

\subsection{Upper Bounds}

Now it suffices to prove the upper bounds. For $\mathsf{OR}$ and $\mathsf{MAX}$, we already know from Theorem~\ref{thm:or_upper} that $\mathcal{O}(K\log(1/\delta)/(1-H(p)))$ is an upper bound for fixed-length setting when comparing two elements with error probability at most $\delta$ requires $\lceil {C\log(1/\delta)}/{(1-H(p))}\rceil $ samples for some constant $C$. Now for variable-length setting, we know  that comparing two elements with error probability at most $\delta$ only requires  $\log(1/\delta)/D_{\mathsf{KL}}(p\|1-p)+1/(1-2p)$ queries in expectation via the variable-length comparison algorithm in Lemma 13 of~\citet{gu2023optimal}.  Thus in Algorithm~\ref{alg:tournament_variable}, we  replace the repetition-based comparisons in Algorithm~\ref{alg:tournament} with the new variable-length comparison algorithm. This  gives the query complexity $\mathcal{O}(K/(1-H(p))+K\log(1/\delta)/D_{\mathsf{KL}}(p\|1-p))$.
Similar algorithm can also be applied to compute \maxn\!.  The main difference is that in $\mathsf{MAX}$, we directly compare two elements instead of finding the number of $1'$s of the first element.

\begin{theorem}\label{thm:or_upper_variable}
The expected number of total queries made by Algorithm~\ref{alg:tournament_variable} is upper bounded by $C \cdot \left(\frac{K}{1-H(p)}+\frac{K\log(1/\delta)}{D_{\mathsf{KL}}(p\|1-p)}\right)$. Furthermore, the algorithm outputs the correct value of $\mathsf{OR}(X_1,\cdots,X_K)$ with probability at least $1-2\delta$ when $\delta<1/2$.
\end{theorem}

\begin{proof}
First, we compute the total number of queries of the algorithm. Without loss of generality, we may assume that $K$ can be written as $2^m$ for some integer $m$. If not we may add no more than $K$ extra dummy $0$'s to the original list $\mathcal{X}$ to make sure $K=2^m$. In each of the outer iteration $i$, the size of $\mathcal{X}$ is decreased half. We know that after $\lceil \log_2(K)\rceil$ iterations, the set $\mathcal{X}$ will only contain one element. From Lemma 27 in~\citet{gu2023optimal}, the expected number of queries we make is $\mathcal{O}(\log(1/\tilde\delta_i)/D_{\mathsf{KL}}(p\|1-p)+1/(1-2p))$ at round $i$. Thus the total number of queries we make is
\begin{align*}
   & \frac{\log(1/ \delta )}{D_{\mathsf{KL}}(p\|1-p)}+\frac{1}{1-2p} + C\cdot\sum_{i=1}^{\lceil \log_2(K)\rceil}\left(\frac{\log(1/\tilde\delta_i)}{D_{\mathsf{KL}}(p\|1-p)}+\frac{1}{1-2p}\right)\cdot \frac{K}{2^{i-1}} \\
  \leq   & C\cdot \sum_{i=1}^{\lceil \log_2(K)\rceil}\left(\frac{(4i-2)\log(1/\delta)}{D_{\mathsf{KL}}(p\|1-p)}+\frac{1}{1-2p}\right)\cdot \frac{K}{2^{i-1}} \\
    \leq & C \cdot\left( \frac{K}{1-2p} + \frac{K\log(1/\delta)}{D_{\mathsf{KL}}(p\|1-p)}\cdot \sum_{i=1}^{\lceil \log_2(K)\rceil}\frac{4(2i-1)}{2^{i-1}} \right)\\
   \leq &   C \cdot\left( \frac{K}{1-2p} + \frac{K\log(1/\delta)}{D_{\mathsf{KL}}(p\|1-p)} \right)\\ 
   \leq & C \cdot\left( \frac{K}{1-H(p)} + \frac{K\log(1/\delta)}{D_{\mathsf{KL}}(p\|1-p)} \right).
\end{align*}
Here in last inequality we use the fact below, which can be verified numerically: 
\begin{align}
\forall p\in[0, 1],    \frac{(1/2-p)^2}{1-H(p)} \in [1/4,1/2]. 
\end{align}

Now we show that the failure probability of the algorithm is at most $\delta$. Consider the first case where all the elements are $0$. Then no matter which element is left in $\mathcal{X}$, the probability that the algorithm fails is the probability that the last while loop gives wrong output. From Lemma 27 in~\citet{gu2023optimal}, we know that such probability is less than $\delta$. 

Consider the second case where there exists at least a $1$ in the original elements $X_1,\cdots, X_K$. Without loss of generality, we assume that $X_1=1$. Let $\mathcal{X}^i$ be the remaining list of elements at the beginning of $i$-th iteration. We let $\mathcal{E}_i$ be the event that the first element in $\mathcal{X}^i$ is $1$ while the first element in $\mathcal{X}^{i+1}$ is $0$. This event only happens when the while loop ends with $a<\tilde\delta_i$ for the first element at the $i$-th round. Let $\mathcal{A}$ denote the event that the only left element is $1$ in the last round,  we have 
\begin{align*}
    \mathbb{P}(\mathcal{A}) & \geq  1 - \mathbb{P}\left(\bigcup_{i=1}^{\lceil \log_2(K)\rceil} \mathcal{E}_i\right)     \\
    & \geq 1-\sum_{i=1}^{\lceil \log_2(K)\rceil} \mathbb{P}(\mathcal{E}_i) \\  
    & >1-\sum_{i=1}^{\lceil \log_2(K)\rceil} \delta^{2(2i-1)} \\ 
    & > 1- \frac{\delta^2}{1-\delta^2}.
\end{align*}
We can upper bound the error probability under event $\mathcal{A}$ with $\delta$. Thus the total failure probability is upper bounded by $\delta + \delta^2/(1-\delta^2)<2\delta$ when $\delta<1/2$.

\end{proof}

The  upper bound for $\mathsf{SEARCH}$ is due to~\citet{gu2023optimal}, and is thus omitted here. The upper bound for $\mathsf{SORT}$  is based on 
that for $\mathsf{SEARCH}$. We can design an insertion-based sorting algorithm by adding elements sequentially to an initially empty sorted set via noisy searching, as in Algorithm 1 in~\citep{isit2022paper}.  Since the insertion step requires $\mathcal{O}(\log(K)/(1-H(p))+\log(1/\delta)/D_{\mathsf{KL}}(p\|1-p))$ queries. Overall we know that one needs $\mathcal{O}(\sum_{k=1}^K (\log(k)/(1-H(p))+\log(1/\delta)/D_{\mathsf{KL}}(p\|1-p)))$ $=\mathcal{O}(K\log(K)/(1-H(p))+K\log(1/\delta)/D_{\mathsf{KL}}(p\|1-p))$ queries to achieve error probability at most $K\delta$. Rescaling $\delta$ gives the final rate. 
\end{document}